\documentclass[preprint,12pt]{elsarticle}

\usepackage{graphicx}
\usepackage{epstopdf}

\usepackage{amssymb}
\usepackage{amsmath}
\newtheorem{thm}{Theorem}
\newtheorem{cor}{Corollary}
\newtheorem{prop}{Proposition}
\newtheorem{prob}{Problem}
\newtheorem{lem}[thm]{Lemma}
\newdefinition{rem}{Remark}
\newdefinition{defi}{Definition}
\newproof{proof}{Proof}

\journal{*****}

\begin{document}

\begin{frontmatter}

\title{Integrability of the one dimensional Schroedinger equation}
\author{Thierry COMBOT\fnref{label2}}
\ead{thierr.combot@u-bourgogne.fr}
\address{}

\title{Integrability of the one dimensional Schroedinger equation}

\author{}

\address{}

\begin{abstract}
We present a definition of integrability for the one dimensional Schroedinger equation, which encompasses all known integrable systems, i.e. systems for which the spectrum can be explicitly computed. For this, we introduce the class of rigid functions, built as Liouvillian functions, but containing all solutions of rigid differential operators in the sense of Katz, and a notion of natural of boundary conditions. We then make a complete classification of rational integrable potentials. Many new integrable cases are found, some of them physically interesting.
\end{abstract}

\end{frontmatter}

\section{Introduction}

In this article, we are interested in the definition of integrability and the search of integrable potentials of the one dimensional Schroedinger equation
\begin{equation}\label{eq1}
\frac{d^2\psi}{dz^2} +(V(z)+E)\psi(z)=0
\end{equation}
where $V$ is a rational function, and $E$ a parameter. The problem is not only to find solutions of equation \eqref{eq1} under a more or less explicit form, but above all to compute the set $\mathcal{S}$ (called the spectrum) parameters values $E$ such that equation \eqref{eq1} admits a solution with particular properties (called boundary conditions). The most typical condition required is the square integrable condition
$$\int\limits_{-\infty}^{\infty} \mid \psi(z) \mid^2 dz <\infty$$
Equation \eqref{eq1} is the quantum equivalent of a one degree of freedom Hamiltonian system. In classical mechanics, this system is always integrable in the sense that we are always able to express the solutions in terms of quadrature. In this quantum equivalent, it is no longer the case.

There has been various ways to define the meaning of $V$ quantum integrable. Here we focus on complete integrability, i.e. finding all the eigenfunctions, which is more restrictive to partial integrability cases, as in \cite{3} where only finitely many eigenstates are found. The closest notion to the one presented in this article is the following

\begin{defi}[Integrability definition in \cite{1}]
The equation \eqref{eq1} is said to be integrable if for all $E$ in the spectrum, the solutions of equation \eqref{eq1} are Liouvillian.
\end{defi}

This definition seems to contain all ``quantum integrable'' cases, at least in the case of discrete spectrum. However, there are several inconvenience
\begin{itemize}
\item This does not allow (at least a priori) to compute the set $S$. Indeed, the integrability check has to be done the other way. Assume we know the spectrum $S$, the system is integrable if and only if equation \eqref{eq1} has Liouvillian solutions. This can however, for a fixed $E$, be done algorithmically through the Kovacic algorithm \cite{2}.
\item The notion of integrability is strongly dependant of the boundary conditions. This happens for example in the following equation
$$\frac{d^2\psi}{dz^2} +(z^{-1}+E)\psi(z)=0$$
If the boundary condition is being square integrable on $\mathbb{R}^+$, then the system is integrable with a discrete spectrum $\mathcal{S}$. If we look for square integrable solutions on $\mathbb{R}^-$ there are none, and if only near $0$ then $\mathcal{S}=\mathbb{C}$. In this continuous case, the system is not integrable in the above sense.
\item The Liouvillian condition is somewhat arbitrary. In particular, there are other functions which are dubbed ``nice'' but not Liouvillian, see \cite{4,5,6}.
\end{itemize}

So our purpose is to build a definition of integrability which is as most as possible independent of the boundary conditions, which allows to compute in an algebraic manner the spectrum, and which is large enough to contain all cases dubbed to be ``quantum integrable''. For this we construct a class of functions in section $2$, we name ``rigid functions'', the name coming from the notion of rigid operators introduced by Katz \cite{8}, to which they are closely linked.

\begin{defi}\label{def1}
A potential $V\in\mathbb{C}(z)$ is said to be quantum integrable if for all $E\in \mathbb{C}$, the solutions of equation \eqref{eq1} are rigid functions.
\end{defi}

In section $2$ is also introduced a notion of natural boundary conditions. The boundary conditions are natural if they can be expressed in terms of monodromy and Stokes matrices, see Definition \ref{defnat}. We prove in particular that the classical square integrability condition is ``almost'' equivalent to a natural boundary condition, in the sense that there exists a natural boundary condition which gives an infinite discrete set of energies containing those for which the square integrability condition is satisfied (as long as this one is not always satisfied).
We will see moreover how to compute explicitly the spectrum from the expressions of solutions in terms of rigid functions and boundary conditions. The quantum integrable potentials split naturally in two categories. The discrete type, for which there exist natural boundary conditions leading to an infinite countable spectrum, and the continuous one, for which any natural boundary condition leads to finite or continuous spectrum. The latter is related to isomonodromic deformations. In section $3$, we prove Theorem \ref{thmmain0}, the quantum integrable potentials $V$ should have eigenfunctions of $4$ possible forms. In section $4$, we prove the main theorem of the article, a classification of quantum integrable rational potentials. The families of integrable potentials are generated by Pade interpolation/series. The section $5$ and the Appendix is devoted to examples and the explicit generation of the quantum integrable potentials of these families, as their presentation uses Pade interpolation and Pade series which makes their construction not immediate, although straightforward. Among them, two physically interesting new quantum integrable potentials are solved in details
$$V(z)=-z^2-2-\frac{8}{2z^2+1}+\frac{16}{(2z^2+1)^2}$$
$$V(z)=\frac{1}{z}-\frac{4}{z^2+2z+2}+\frac{8}{(z^2+2z+2)^2}$$

For the rest of the article, we will note $\mathcal{W}(\mu,\nu,z)$ a non zero solution of the differential equation
$$y''(z)+\left(-\frac{1}{4}+\frac{\mu}{z}+\frac{1/4-\nu^2}{z^2}\right) y(z)=0.$$
Moreover, from now on, the $'$ will be the differentiation in $z$.\\

\begin{thm}\label{thmmain0}
If $V\in\mathbb{C}(z)$ is quantum integrable, then up to affine coordinate change and addition of a constant to $V$ the Schroedinger equation has solutions of one of the following forms
\begin{align*}
\psi(z,E) & =\frac{z^{3/2}\left(\frac{M(z,E)}{2z}\mathcal{W}(E/4,\nu,z^2)+\mathcal{W}'(E/4,\nu,z^2)\right)}{\sqrt{M(z,E)^2z^2+M(z,E)z-M'(z,E)z^2-z^4+z^2E-4\nu^2+1}}\\
\psi(z,E) & =\frac{z\left(\frac{M(z,E)}{\sqrt{-4E}}\mathcal{W}((-4E)^{-1/2},\nu,z\sqrt{-4E})+\mathcal{W}'((-4E)^{-1/2},\nu,z\sqrt{-4E})\right)}
{\sqrt{4M(z,E)^2z^2+4z^2E-4M'(z,E)z^2-4\nu^2+4z+1}}\\
\psi(z,E) & =\frac{z \left ( \frac{M(z,E)}{\sqrt{-4E}} \mathcal{W}(0,\nu,z\sqrt{-4E})+\mathcal{W}'(0,\nu,z\sqrt{-4E}) \right)}{\sqrt{4M(z,E)^2z^2+4z^2E-4M'(z,E)z^2-4\nu^2+1}}\\
\psi(z,E) & =\frac{\frac{i(-M(z,E)(z+E)/2+1/8)}{(z+E)^{3/2}}\mathcal{W}(0,\frac{1}{3},\frac{4i}{3}(z+E)^{3/2})+\mathcal{W}'(0,\frac{1}{3},\frac{4i}{3}(z+E)^{3/2})}
{(z+E)^{-1/4}\sqrt{M(z,E)^2+E-M'(z,E)+z}}
\end{align*}
with $M(z,E)$ rational in $z,E$.
\end{thm}

The case $M=\infty$ has to be included, and effectively leads to quantum integrable potentials. Remark that given a solution $\psi$ of the Schroedinger equation, we can recover the potential as $V+E=-\psi''/\psi$. Thus the function $M$, and even its restriction to a generic value of $E$, completely defines the potential $V$ in the above expressions. In particular, the potential $V(z)+E$ can be written as a rational function of $z,E,M$ and its derivatives. Remark that however, all rational $M$ do not lead to potentials, as $-\psi''/\psi$ should be of the form $V(z)+E$. This will be the condition to obtain an integrable potential. We now present the classification results, i.e. a set of $M$ functions leading to all quantum integrable potentials $V$.

\begin{thm}\label{thmmain}
A quantum integrable potential $V\in\mathbb{C}(z)$ comes from a function $M$ given by
\begin{itemize}
\item In case $1$ of Theorem \ref{thmmain0}, the rational interpolation with numerator denominators degrees in $E$ less than $n/2,(n-1)/2$ given by
$M(z,\epsilon_1(4k+2)+4\epsilon_2\nu)=$
\begin{equation}\label{eqricsol}
-\frac{\partial}{\partial z} \ln\left( z^{2\epsilon_1\epsilon_2\nu+1}e^{\epsilon_1z^2/2} \!_1F_1(-k,2\epsilon_1\epsilon_2\nu+1,\epsilon_1z^2) \right)
\end{equation}
for $n$ points of the form $\epsilon_1(4k+2)+4\epsilon_2\nu,\; k\in\mathbb{N},\epsilon_1,\epsilon_2=\pm 1$
\item In case $2$ of Theorem \ref{thmmain0}, the rational interpolation with numerator denominators degrees in $E$ less than $n/2,(n-1)/2$ given by
$M(z,-(2\epsilon\nu+2k+1)^{-2})=$
\begin{align}\begin{split}\label{eqricsol2}
-\frac{\partial}{\partial z} \ln\left(z^{\epsilon\nu+1/2}e^{-\frac{z}{2\epsilon\nu+2k+1}} \!_1F_1\left(-k,2\epsilon\nu+1,\frac{2z}{2\epsilon\nu+2k+1}\right) \right)
\end{split}\end{align}
for $n$ points $-(2\epsilon\nu+2k+1)^{-2}$ with $k\in\mathbb{N},\epsilon=\pm 1$.
\item In case $3$ of Theorem \ref{thmmain0} the singular $M=\infty$.
\item In case $3$ of Theorem \ref{thmmain0} with $\nu = 0$, the rational function with numerator denominators degrees in $E$ less than $n/2,(n-1)/2$ defined by the series
$$M(z,E)=-\frac{\partial}{\partial z} \ln\left(\sum\limits_{i=0}^{n-1} D^i F(z) E^{n-1-i} \right) +O(E^n)$$
with $D=-\partial_z^2-1/(4z^2)$, $F(z)=P_1(z^2)+\ln z P_2(z^2)$ and $\deg P_1= n-1, \deg P_2\leq n/2-1$.
\item In case $3$ of Theorem \ref{thmmain0} with $\nu=1/2$, the rational function with numerator denominators degrees in $E$ less than $n/2,(n-1)/2$ defined by the series
$$M(z,E)=-\frac{\partial}{\partial z} \ln\left(\sum\limits_{i=0}^{n-1} (-1)^{i}\partial_z^{2i} F(z) E^{n-1-i}\right) +O(E^n)$$
with $F$ polynomial, $\deg F= 2n-1 \hbox{ or } 2n-2$.
\item In case $4$ of Theorem \ref{thmmain0}, the singular $M=\infty$.
\end{itemize}
\end{thm}

\noindent
\textbf{Remarks}\\
The case $n=0$ (no interpolation points or series) will conventionally give $M=\infty$ (constant infinite function) and this convention allows to recover the potentials $z^2+\alpha/z^2,1/z+\alpha/z^2,z,\alpha/z^2$ which are singular cases in our classification.\\
The interpolation points could be not distinct: for specific values of $\nu$, two interpolations points given by different $k,\epsilon_1,\epsilon_2$ can be equal. The rational interpolation is then given by a limit process when $\nu$ tends to the specific value.\\
The $M$ function used to express a quantum integrable potential is not unique. This is due to recurrence relations between Whittaker functions. This induces a homographic transformation on $M$, and so infinitely many $M$ can give the same potential.\\

\section{Quantum integrability of 1D rational potentials}

As said before, a quantum problem is given by a potential $V\in\mathbb{C}(z)$ and some additional conditions on the solutions we are searching. We want an integrability definition that is as generic as possible, i.e. not depending on these boundary conditions but only to the potential $V$. Still some boundary conditions seem more natural than others. For example, asking that a solution should vanish on some fixed point seem too arbitrary to be acceptable. Indeed, this condition has the physical sense of an infinite wall at an arbitrary point, and so adding such boundary condition corresponds to the transformation $V(z) \rightarrow V(z)+\delta_a(z)$ where $\delta_a$ is a Dirac at $a\in\mathbb{C}$. This can be understood as a modification of the potential (adding a singularity to $V$) more than just a boundary condition for the quantum problem. So we need to restrict ourselves to ``admissible'' boundary conditions.

\subsection{Natural boundary conditions}

\begin{defi}[singularities]
Let us consider a linear differential equation
$$a_n(t)y^{(n)}(t)+ \dots + a_0(t)y(t)=0$$
with $a_i$ polynomials, $a_n\neq 0$ and relatively prime. The roots of $a_n$ are called singularities. If $\alpha$ is not a root of $a_n$, $\alpha$ is called a regular point. At a singularity $\alpha$, if the systems admits a converging Puiseux series (possibly with logs) basis of solutions, the point $\alpha$ is called singular regular, else $\alpha$ is called singular irregular. If moreover these Puiseux series are Laurent series, we call $\alpha$ a meromorphic singularity, and if polynomial series, apparent singularity.
\end{defi}

Near a meromorphic singularity, the solutions of the differential equation are univalued. In the even more special case of an apparent singularity, the point $\alpha$ is not a singularity for any solution of the differential equation (so is the origin of ``apparent'').\\

\begin{defi}[monodromy]\label{defmon}
Let us consider a linear differential equation
$$a_n(t)y^{(n)}(t)+ \dots + a_0(t)y(t)=0$$
with $a_i$ polynomials, $a_n\neq 0$ and relatively prime. Let $\alpha\in \mathbb{C}$ be a regular point, and $B$ a series basis solution at $\alpha$. Let $\gamma \subset \mathbb{C}$ be a closed oriented curve not containing singular points, with $\alpha\in\gamma$. By analytic continuation, we can extend the basis of solution $B$ at $\alpha$ along $\gamma$. After one loop, we obtain a solution basis $B'$. As $B,B'$ are both solutions basis at $\alpha$, there exists a matrix $M_\gamma$ such that $B'=B M_\gamma$ called \textit{the monodromy matrix along $\gamma$}.
\end{defi}

Remark that if the monodromy around a point is trivial, then it is either a regular point or at worst a meromorphic singularity. Indeed, the monodromy around a point encode the local multivaluation of the solutions of the differential equation.

\begin{defi}[Stokes]\label{defsto}
Let us consider a linear differential equation
$$a_n(t)y^{(n)}(t)+ \dots + a_0(t)y(t)=0$$
with $a_i$ polynomials, $a_n\neq 0$ and relatively prime. Let $\alpha\in \mathbb{C}$ be a singular irregular point. We can construct a basis of solutions at $\alpha$ with formal power series of the form
$$e^{\sum\limits_{i=1}^n c_i (z-\alpha)^{-i/p}} (z-\alpha)^\gamma \ln(z-\alpha)^k \sum\limits_{i=0}^\infty b_i (z-\alpha)^i$$ 
Such formal series solution can be identified with a Gevrey function solution of the differential equation following a particular direction towards $\alpha$ except for finitely many directions, called singular directions,and the directions between them called sectors. For each sector, the formal basis is identified to a Gevrey function basis, and going from one sector to the next defines a basis change, \textit{the Stokes matrix}. The monodromy matrix generated by the truncated formal series solution along a small loop around $\alpha$ is called the \textit{formal monodromy}. The monodromy matrix defined as in definition \ref{defmon} above along a small loop around $\alpha$ is called the \textit{true monodromy}.
\end{defi}

Remark that along a small loop around a singular regular point, it is easy to compute the monodromy matrix. The difficulty appear when $\gamma$ encompass several singularities. The path $\gamma$ can be deformed, but still we need to know how a Puiseux series solutions at one singularity reconnects with another at the other singularity. Let us now define the notion of natural boundary conditions.

\begin{defi}[Natural boundary conditions]\label{defnat}
Let us consider equation \eqref{eq1}. A natural boundary condition on solutions of \eqref{eq1} can be written under the form
$$\prod\limits_{i=1}^m M_i^{w_i} \in J$$
where $M_i$ are Stokes or monodromy matrices, $w_i\in \mathbb{Z}$ and $J$ a set of conjugacy class of matrices.
\end{defi}

The product encodes a path with integer turns around singularities and integer many crossing of singular directions. The fact that the condition has to be about a conjugacy class of matrices and not equal to a particular matrix is because of the arbitrary initial basis choice. Indeed, we have fixed a common point and basis arbitrary, and so if we want to get rid of this arbitrary choice, we need to consider that the matrices $M_i$ are defined up a \textbf{common} basis change:
$$(M_1,\dots, M_m) \rightarrow (P^{-1}M_1P,\dots, P^{-1}M_m P) \qquad P\in GL_2(\mathbb{C})$$

\begin{prop}
The spectrum for equation \eqref{eq1} with natural boundary conditions is the set of roots of
\begin{itemize}
\item a holomorphic function $f(E)$ if $\lim\limits_{z=\infty} V(z)=\infty$
\item a holomorphic function $f(\ln E), E\in\mathbb{C}^*$ if $\lim\limits_{z=\infty} V(z)=0$
\end{itemize}
\end{prop}

Remark that if $V$ converges at infinity, we can always assume it converges to $0$ as we can always make an energy shift for $E$. 

\begin{proof}
Equation \eqref{eq1} comes with a parameter $E$, which plays a fundamental role. The monodromy and Stokes matrices depend a priori on this parameter. Let us first remark that the singularities of \eqref{eq1} do not move with respect to $E$. The same applies for singular directions, except possibly at infinity for $E=0$ for which singular direction crossing is possible: indeed, the asymptotic behaviour of solutions change when $E=0$ if $V$ tends to zero at infinity.

Thus when $\lim_{z=\infty} V(z)=\infty$, the monodromy and stokes can be globally defined on $\mathbb{C}$ as functions of $E$. And as our equation depends analytically on $E$, all these matrices are holomorphic functions of $E$. The natural boundary conditions are put on this matrices, and so this gives the first case of the corollary.

When $\lim_{z=\infty} V(z)=0$, the monodromy and stokes are defined on $\mathbb{C}^*$. However, $\mathbb{C}^*$ is not simply connected, and thus this does not immply that these matrices are globally defined on $\mathbb{C}^*$. We need to consider the universal covering of $\mathbb{C}^*$. Our equation depends analytically on $E$, all these matrices are locally holomorphic functions of $E$. Locally holomorphic functions on the universal covering of $\mathbb{C}^*$ are holomorphic function in $\ln E$, and thus the corollary follows.
\end{proof}\qed

Let us now remark that the matrices $M_i$ of Definition \ref{defnat} depend on $E$, and they are not well defined for $E=0$ when $\lim_{z=\infty} V(z)=0$. Thus in this case the problem to know whether $E=0$ belongs to the spectrum is not defined through this presentation of natural boundary conditions. This can be explicitly seen on the example $V(z)=1/z$, for which the solutions of the Schroedinger equation are
$$\psi(z,E)=\mathcal{W}\left(-\frac{i}{2\sqrt{E}},\frac{1}{2},2i\sqrt{E}z\right)$$
The case $E=0$ is a singularity of this equation, as singular directions at the irregular point $\infty$ are crossing. The Whittaker function simplifies in the Bessel function. So from now on, this problem will be skipped completely by assuming that $E\in\mathbb{C}^*$ when $V(z)$ converges (and then assuming it converges to $0$).

\begin{defi}
Let us consider Schroedinger equation \eqref{eq1} with $\lim_{z=\infty} V(z)=0$ if $V$ converges. We say that $V$ is of continuous type if, up to common basis change, the monodromy and Stokes matrices do not depend on $E\in \mathbb{C}^*$. Else the equation is said of discrete type.
\end{defi}

Said otherwise, in the continuous case, the transformation $V(z) \longrightarrow V(z)+\epsilon$ is an isomonodromic iso-Stokes deformation. Such kind of deformations are very rare, and have been analysed by Painleve, leading to the so called Painleve equations \cite{9}.

\subsection{The square integrability condition}

The most classical boundary condition is square integrability of one solution
$$\int_{\mathbb{R}} \mid \psi \mid^2 dz <\infty$$
Although this seems to be a global condition (and so the word boundary would be inappropriate), the solutions $\psi$ are always regular outside the singularities of the differential equation. So the condition of square integrability comes down to analysing the behaviour at singularities. For our definition of natural boundary condition to be reasonable, it should include this square integrability condition. This will not always be exactly the case, so let us define a little larger notion

\begin{defi}
Let us consider equation \eqref{eq1} with $\lim_{z=\infty} V(z)=0$ if convergent, and some boundary conditions. Let us note $\mathcal{S}$ the set of $E\in\mathbb{C}^*$ satisfying these boundary conditions. We say that these boundary conditions are almost natural if there exist natural boundary conditions, defining a set $\mathcal{C}$, and such that
$$\mathcal{S} \subset \mathcal{C} \qquad \dim \mathcal{S} =\dim \mathcal{C}$$
\end{defi}

In this definition, an almost natural condition is ``close'' to a natural boundary condition in the sense that if $\mathcal{S}$ is discrete infinite (the case with physical sense), then one can find a set $\mathcal{C}$ containing it which is also discrete infinite. With the set $\mathcal{C}$, the ``structure'' of the spectrum $\mathcal{S}$ is known, we just have to remove some ``errors''.

\begin{defi}
Let us consider $\alpha$ a real singularity of the Schroedinger equation \eqref{eq1}. We say that the singularity $\alpha$ is active if the space of formal series solutions at $\alpha$ contains exactly a subspace of dimension $1$ of square integrable near $\alpha$ formal series.
\end{defi}

\begin{prop}\label{proptri}
We consider the Schroedinger equation \eqref{eq1} with $\lim_{z=\infty} V(z)=0$ if $V$ converges. Let us consider a fixed based point and a series solution basis at this point. For all $\alpha_i\in\mathbb{R}$ active singularities of equation \eqref{eq1}, consider the Stokes matrices going from direction $\mathbb{R}^{-}$ to $\mathbb{R}^+$ and the true monodromy matrices, and denote $G$ the multiplicative group generated. Let us note
$$\mathcal{S}=\{E\in\mathbb{C}^*,\exists \psi \hbox{ solution of \eqref{eq1} with }\int_{\mathbb{R}} \mid \psi \mid^2 dz <\infty \}$$
$$\mathcal{C}_1=\{E\in\mathbb{C}^*, G \hbox{ cotriangularizable} \}$$
$$\mathcal{C}_2=\{E\in\mathbb{C}^*, G \hbox{ codiagonalizable} \}$$
$$\mathcal{C}_3=\{E\in\mathbb{C}^*, G \subset I_2\}$$
Assume $\mathcal{S}\neq \mathbb{C}^*,\emptyset$. Then there is one $\mathcal{C}_i$ discrete countable such that $\mathcal{S} \subset \mathcal{C}_i$.
\end{prop}

Remark that the $\mathcal{C}_i$ come from natural boundary conditions. But we have not a priori $\mathcal{S}=\mathcal{C}_i$ for some $i$. This is however the case for typical quantum integrable physical systems as the potentials $V(z)=z^2,1/z$. If $\mathcal{S}$ is infinite countable (which is typically the interesting case), then the inclusion $\mathcal{S}\subset \mathcal{C}_i$ is strong, and so we can say that the square integrable condition is almost equal to a natural boundary condition as one of the $\mathcal{C}_i$ contains $\mathcal{S}$ and is of same dimension.\\

\begin{cor}
The square integrability condition is almost natural.
\end{cor}

\begin{proof}
Let us consider an $E\in \mathcal{S}$ and $\alpha\in\mathbb{R}$ be a singularity. In the general irregular case with $\alpha\in\mathbb{R}$, we have such kind of behaviour
$$e^{\sum\limits_{i=1}^m c_i (z-\alpha)^{-i/p}} (z-\alpha)^\gamma \ln(z-\alpha)^k,\qquad p\in\mathbb{N}^*,k\in\{0,1\}, \; c_i,\gamma\in\mathbb{C}$$
A basis of such formal solutions lives in a differential field extension over the field of Laurent series, and thus we can attach to it a differential Galois group. Remark that the Schroedinger equation is unimodular, the Wronskian is constant, and so this group is in $SL_2(\mathbb{C})$. As it is diagonal, it is generated by one matrix, we will note in this proof $M_\alpha$. We can moreover assume that $M_\alpha$ is in $G$, as $M_\alpha$ belong to the local differential Galois group at $\alpha$ (i.e. the Galois group over the base field of meromorphic functions on an open neighbourhood of $\alpha$).

Depending on the parameters, these formal series can be square integrable or not. Let us denote $E_{int}$ the subspace of formal series square integrable. The dimension of this space can be $1,2$. As there is a square integrable solution, there is a non zero element of $E_{int}$ which is sent to $E_{int}$ be the Stokes matrix $S$ going from direction $\mathbb{R}^{-}$ to $\mathbb{R}^+$. And the same for the true monodromy. Let us remark that this is automatically satisfied if $E_{int}$ is of dimension $2$, so we can restrict ourselves to the case of dimension $1$, for which $\alpha$ is called an active singularity. This has to be satisfied simultaneously for all real singularities, and so all true monodromy matrices and Stokes matrices from $\mathbb{R}^{-}$ to $\mathbb{R}^+$ at active singularities have to stabilize a common vector space. In other words, the group $G$ has to be cotriangularizable. Thus $\mathcal{S} \subset \mathcal{C}_1$.

If $\mathcal{C}_1$ is discrete countable, then the proposition is proved. So we can now assume that $\mathcal{C}_1=\mathbb{C}^*$. In other words, the group $G$ is triangular for all $E$. But Proposition \ref{proptri} has the hypothesis $\mathcal{S}\neq \mathbb{C}^*$ and so the vector space $E_{int}$ is not always stabilized by $G$. Let us now remark that the matrix $M_\alpha$ always stabilizes the vector space $E_{int}$. And this matrix at an active singularity is in the group $G$.

Let us first assume that there exist an active singularity such that $M_\alpha$ is not identity. Now two cases:
\begin{itemize}
\item either $M_\alpha$ is diagonalizable with distinct eigenvalues. Then after to basis change, we can assume that $M_\alpha$ is diagonal and the group $G$ is triangular (recall that $M_\alpha\in G$). As $E_{int}$ is of dimension $1$, is not generated by the vector $(1,0)$ (else it would be stabilized by $G$ for all $E$) and is stabilized by $M$, then
$$E_{int} =\mathbb{C}.(0,1).$$
If $E\in\mathcal{S}$, we have that this vector space is stabilized by $G$. Thus $G$ stabilize two supplementary $1$-dimensional vector spaces, and so is diagonal. This gives $\mathcal{S} \subset \mathcal{C}_2$.
\item either $M_\alpha$ is not diagonalizable. So both eigenvalues of $M_\alpha$ are equal to $1$, and $M_\alpha$ is triangular (after basis change). However, $M_\alpha$ stabilize $E_{int}$, which can only be $\mathbb{C}.(1,0)$. But then $E_{int}$ is stabilized by $G$ for any $E$, and this would implies $\mathcal{S}= \mathbb{C}^*$. Impossible.
\end{itemize}
The last remaining case is when all matrices $M_\alpha$ at active singularities are identity. Then the formal series solutions cannot have nor exponentials, nor fractional/irrational powers. So all active singularities are meromorphic singularities (formal series solutions are Laurent series), and thus the group $G$ is reduced to identity, i.e. $\mathcal{C}_2=\mathbb{C}^*$.

If $\mathcal{C}_2$ is discrete countable, then the proposition is proved. So we can now assume that $\mathcal{C}_2=\mathbb{C}^*$. In other words, the group $G$ is diagonal for all $E$. But Proposition \ref{proptri} has the hypothesis $\mathcal{S}\neq \mathbb{C}^*$ and so the vector space $E_{int}$ is not always stabilized by $G$. The matrices $M_\alpha$ at active singularities are also diagonal, and stabilize the vector space $E_{int}$. However this cannot be $\mathbb{C}.(1,0), \mathbb{C}.(0,1)$, as else it would be stabilized by $G$ for all $E$. So these matrices must have a third stable vector space. And so they are identity. This implies that the group $G$ is reduced to identity. So any vector space is stabilized by $G$, and thus in particular $E_{int}$. Thus $\mathcal{S}\subset \mathcal{C}_3$. Finally, if $\mathcal{C}_3=\mathbb{C}^*$, the group $G$ is identity for all $E$, and thus $\mathcal{S}$ is either $\emptyset$ or $\mathbb{C}^*$.
\end{proof}\qed

For the computation of the spectrum of quantum integrable system under the square integrability condition, we will first compute the sets $\mathcal{C}_i$ which can be found algebraically from the monodromy/Stokes matrices. Then we obtain a countable discrete set of ``candidates''  and we can look more precisely the behaviour of solutions at singularities to check which energies satisfy the square integrability condition.\\

\noindent
\textbf{Other examples of almost natural boundary condition}
\begin{itemize}
\item Prescribed singular behaviour at one singularity
\item Prescribed radius of convergence for series solutions
\item Analyticity of solutions of a particular domain
\item To belong to the Bargman space (holomorphic functions with $\mid f(z) \mid^2\exp(-z^2/2)$ integrable)
\item Prescribed ramification/coverings of the Riemann sphere
\end{itemize}

\subsection{Rigid functions}

Our objective now is to compute the spectrum, and more precisely defined a class of Schroedinger equation \eqref{eq1} for which the monodromy and Stokes matrices can be explicitly computed. At first view, this seems to be intractable, as analytic continuation of formal series is used everywhere to define these matrices. However, the idea of rigid operators \cite{8} is to gather all algebraic information we have at our disposal, and try to compute these matrices.

The main information we have is local monodromy in the singular regular case (so a small loop around a singularity) and the formal monodromy matrix for irregular singular points. Making a whole turn around this irregular singular point gives moreover a multiplicative relation between formal monodromy (known), true monodromy (unknown), and Stokes matrices (unknown). Finally we have a global structure: if we have $n$ singularities, making a turn around $n-1$ of them equals to making a turn around the one left (recall that we are on the Riemann sphere). However, these matrices are not known in a common basis.

\subsubsection{Regular case}

Let us first focus on the regular singular case, i.e. no irregular singularities at all. To summarize, we search $M_1,\dots,M_m\in GL_n(\mathbb{C})$ such that
$$M_1\dots M_{m-1}=M_m^{-1}$$
just knowing the $M_i$ up to conjugacy. Of course, one just has to find the $M_1,\dots,M_m$ up to \textbf{common} basis change. Is this enough to find the $M_i$? Sometimes \cite{14}

\begin{defi}\label{defrigid1}
Let us consider a linear differential operator
$$y(t) \longrightarrow a_n(t)y^{(n)}(t)+ \dots + a_0(t)y(t)$$
with $a_i$ polynomials, $a_n\neq 0$ relatively prime, and with only regular singularities. The operator is said to be rigid if the monodromy matrices are uniquely defined up to \textbf{common} basis change by their conjugacy class given by local monodromy.
\end{defi}

Search for rigid operators is still on going, and is known as the Deligne-Simpson problem \cite{12}:
\begin{prob}
Let $n,m$ be two positive integers. Find all $m$-uplet $(G_1,\dots G_m)$ of conjugacy classes of $GL_n(\mathbb{C})$ such that the equation
$$M_1\dots M_m=I_n \quad M_i \in G_i,\; i=1\dots m$$
admits a unique solution up to common conjugacy.
\end{prob}

For small dimensions, the problem is solved, and in particular for $n=2$, the only possible solution leads to a famous operator, Gauss hypergeometric differential equation.

\subsubsection{Irregular case}

The definition of rigid operators in the irregular case can be done through a limiting process, the confluence.
Using the parameters in a family of regular rigid operator, we make two singularities fuse with simultaneous rescaling. This leads to an equation with one less singularity, but irregular. Moreover, we have
\begin{itemize}
\item The limit direction when the fusion occurs becomes a singular direction
\item The Stokes matrices are the limits of the monodromy matrices around each singularity of the fusion
\item Multiple singularities can fuse simultaneously, leading to several singular direction and Stokes matrices
\end{itemize}

So the irregular generalisation of rigid operators is straightforward
\begin{defi}
Let us consider a linear differential operator
$$y(t) \longrightarrow a_n(t)y^{(n)}(t)+ \dots + a_0(t)y(t)$$
with $a_i$ polynomials, $a_n\neq 0$ relatively prime. The operator is said to be rigid if either it is regular and rigid according to Definition \ref{defrigid1}, or is produced by a limit confluence process of a family of regular and rigid operators according to Definition \ref{defrigid1}.
\end{defi}

As the order is conserved by the limiting process, we only have to look at confluence processes for the Gauss hypergeometric equation. Outside elementary functions, this produces the Whittaker differential equation (and Bessel differential equation as a specialization).

\subsubsection{Galois rigidity}

There an additional global algebraic structure we have not used yet to compute our monodromy/Stokes matrices, which is the differential Galois group. In particular, due to Ramis theorem, we know that monodromy/Stokes matrices belong to the differential Galois group. This group can be computed algebraically and automatically thanks to the Kovacic algorithm. It is an algebraic Lie subgroup of $GL_n(\mathbb{C})$. So this suggests the following generalization of the Deligne Simpson problem (already raised in \cite{12})

\begin{prob}
Let $n,m$ be two positive integers and $G$ an algebraic Lie subgroup of $GL_n(\mathbb{C})$. Find all $m$-uplet $(G_1,\dots G_m)$ of conjugacy classes of $G$ such that the equation
$$M_1\dots M_m=I_n \quad M_i \in G_i,\; i=1\dots m$$
admits finitely many solutions in $G$ up to common conjugacy.
\end{prob}

\begin{defi}[Galois rigid operators]
Let us consider a linear differential operator
$$y(t) \longrightarrow a_n(t)y^{(n)}(t)+ \dots + a_0(t)y(t)$$
with $a_i$ polynomials, $a_n\neq 0$ relatively prime. The operator is said to be rigid if it is
\begin{itemize}
\item either regular and if the monodromy matrices are defined up to a finite choice up to \textbf{common} basis change by their conjugacy class given by local monodromy and their inclusion to the differential Galois group.
\item a limit confluence process of a family of regular and rigid operators of the above case.
\end{itemize}
\end{defi}

\begin{thm}\label{thmrigid}
The Galois rigid operators of order $2$ without meromorphic singularities are, up to hyperexponential multiplication and Moebius transformation
\begin{itemize}
\item The hypergeometric equation.
\item The Whittaker equation.
\item The logarithm equation $zy''+ y'$.
\item Any operator with dihedral Galois group over $\mathbb{C}(z)$ and diagonal Galois group over $\mathbb{C}(\sqrt{z})$.
\item Any operator with diagonal Galois group or finite Galois group.
\end{itemize}
\end{thm}

\begin{proof}

Let us treat each possible differential Galois group.

If $\hbox{Gal}=GL_2(\mathbb{C}) \hbox{ or } SL_2(\mathbb{C})$, this has already been done in \cite{14,12} for regular operators. The only possible case is the Gauss hypergeometric equation. Its confluence gives the Whittaker equation.

A triangular group. We are first looking for Fuchsian equations. For each matrix, we know its conjugacy class, and so in particular its eigenvalues. By multiplying the solutions by a hyperexponential function, we can fix one one these eigenvalues to $1$ for each matrix $M_i$, and we choose the upper left one. As all the matrices should be upper triangular, only one coefficient is still unknown $\beta_i$
$$M_i=\left(\begin{array}{cc}  1&\beta_i\\0&\lambda_i\end{array} \right) $$
If the two eigenvalues are equal, then $\lambda_i=1$ and we moreover know if $\beta_i$ is $0$ or not. The diagonalizable case with double eigenvalue case leads only to the identity matrix, and so a meromorphic singularity, which is forbidden. The only information known on the $M_i$ is the multiplicative relation $M_1\dots M_m=I_2$. So the relation becomes
$$ \prod\limits_{i=1}^m M_i=I_2$$
The upper right coefficient of this product is
\begin{equation}\label{eq2}
\sum\limits_{i=1}^m \left(\prod\limits_{j=i+1}^m \lambda_j\right) \beta_i
\end{equation}
The other coefficients of the product do not give us additional information. This upper right coefficient is a linear form in the $\beta_i$, the unknowns. As we have removed the case for which $\beta_i$ is known to be zero, this linear form is the only relation we have on the $\beta_i$. Remark now that we are searching the matrices $M_i$ up to common basis change. Here we have to keep the triangular form, so we can make a triangular basis change. Such basis change multiply by some constant the upper right coefficient of the $M_i$. So the rigidity problem comes down to find when equation \eqref{eq2} has finitely many solutions up to transformation
$$(\beta_i)_{i=1\dots m} \longrightarrow  (\alpha \beta_i)_{i=1\dots m} \qquad \alpha\in\mathbb{C}^*$$
So this encodes a projective plane of dimension $m -2$. This has finitely many points if and only if $m=2$.

Thus we have at most $2$ singularities. If there is only one, then it should be meromorphic, as a multivalued function has at least two ramification points on the Riemann sphere. So it has exactly $2$ singularities, and using Moebius transformation, we can fix one of them at $0$, and the other one at infinity. As the differential equation is Fuchsian, the differential Galois group is the Zariski closure of the monodromy matrices group. So in particular this group has a common eigenvector of eigenvalue $1$. And thus a rational solution. So up to multiplication by a rational function, this solution can be sent to $1$, and we obtain that our equation writes down
$$\frac{d^2y}{dz^2}-a(z) \frac{dy}{dz}=0$$
with $a\in\mathbb{C}(z)$. The rational function can only have one pole at zero, and $0$ should be singular regular. So $a(z)=c z^{-1}$. The solutions of this equation can be written
$$y(z)=\int z^c dz$$
If $c\neq -1$, then we obtain a hyperexponential solution, and thus the Galois group is diagonal, which is included in another case. The case $c=-1$ gives the differential equation of the theorem. Remark to conclude that this equation cannot have an irregular confluence.

The dihedral case. The Galois group is not connected. If the projective Galois group is not finite (next case of the theorem), then there are exactly two components. These can be written
$$\left(\begin{array}{cc}  \lambda&0\\0&\mu\end{array} \right), \left(\begin{array}{cc}  0&\lambda\\\mu&0\end{array} \right) $$
In the first case, knowing the eigenvalues allows to determine the matrix. In the second case, this only gives us the product $\lambda\mu$. By multiplying the solutions by a hyperexponential, we can assume the Galois group being unimodular. And so that the determinant of these matrices should be $1$. This implies that $\lambda\mu=-1$ in the second case. Then for each monodromy matrix in the second component, we have one unknown. In the other hand, monodromy matrices in the identity component are fully known.

The multiplicative relation can be written
$$\left(\begin{array}{cc}  \lambda_1&0\\0&\lambda_1^{-1}\end{array} \right) \left(\begin{array}{cc}  0&\mu_1\\-\mu_1^{-1}&0\end{array} \right)\left(\begin{array}{cc}  \lambda_2&0\\0&\lambda_2^{-1}\end{array} \right) \dots \left(\begin{array}{cc}  \lambda_p&0\\0&\lambda_p^{-1}\end{array} \right) \left(\begin{array}{cc}  0&\mu_p\\-\mu_p^{-1}&0\end{array} \right)=I_2$$
We simplified directly the relation by multiplying successive diagonal matrices. In the equation, the $\lambda$'s are known, the $\mu$'s are unknown. Remark that we need an even number of non diagonal matrices, i.e. $p$ even. Writing down the product, we obtain just one equation from identification
$$\prod\limits_{i=1,\; i\hbox{ odd}}^p \lambda_i\mu_i=\prod\limits_{i=1,\; i\hbox{ even}}^p \lambda_i\mu_i$$
As the monodromy matrices are searched up to common basis change, using a diagonal basis change (as we need to keep the diagonal structure of the identity component), we can multiply the $\mu_i$'s by an arbitrary complex number. Still, even up to a multiplication of all the $\mu_i$'s, this equation has finitely many solutions only if $p=2$.

We also know that the differential Galois group is dihedral, and so the solutions can be written under the form
$$y(z)=Ae^{\int f(z)+\sqrt{g(z)} dz}+Be^{\int f(z)-\sqrt{g(z)} dz} \qquad f,g\in\mathbb{C}(z)$$
There are exactly two monodromy matrices outside the identity component of the Galois group. The corresponding singularities are root-poles of odd order of $g$. Using a Moebius transformation, we can put them at $0,\infty$, and thus
$$y(z)=Ae^{\int f(z)+\sqrt{z}g(z) dz}+Be^{\int f(z)-\sqrt{z}g(z) dz} \qquad f,g\in\mathbb{C}(z)$$
These are exactly the solutions of operators of order $2$ with rational coefficients whose differential Galois group is diagonal over $\mathbb{C}(\sqrt{z})$.

For the diagonal Galois group case, the eigenvalues of the monodromy matrices are known, they are all codiagonalizable, and thus we know all the monodromy matrices. In the finite projective case, after multiplication by a hyperexponential function, the Galois group becomes finite, and thus the monodromy matrices are known up to a finite choice, as they belong to the Galois group which is finite.

\end{proof}\qed

\subsubsection{The class of rigid functions}

We can now define the class of functions we are interested in. These are built in a same way as Liouvillian functions, except that ``basic'' functions are now solutions of Galois rigid operators.

\begin{defi}[Rigid functions]
The field of rigid functions $\mathcal{R}$ is the smallest differential field with the following properties
\begin{itemize}
\item $\mathcal{R}$ contains the solutions of all Galois rigid operators
\item if $f\in \mathcal{R}$, then for any algebraic function $g$, $f\circ g \in \mathcal{R}$ (algebraic pullbacks)
\end{itemize}
We moreover define field of rigid functions of order $n$, $\mathcal{R}_n$, as the smallest differential field containing the solutions of all Galois rigid operators of order $\leq n$ and all their algebraic pullbacks.
\end{defi}

\noindent
\textbf{Examples}\\

A Liouvillian function can be written in finite terms using integrals, exponentials, and algebraic functions. In our rigid function field, integrals are not allowed. For a Liouvillian function to be rigid, we have to somehow compute these integrals more explicitly. However, this does not imply that these integrals should be elementary functions. Indeed, Liouvillian elementary functions are rigid functions, but Liouvillian rigid functions are not always elementary. This kind of construction has already been suggested by Mark van Hoeij in \cite{15}, where he suggest $(F,O)$ differential fields. The field of base functions $F$ is exactly the same, coming from rigid differential operators, but he adds the integrals in the set of allowed operators $O$, which is excluded here. \\

\noindent
The error function
$$\int e^{z^2} dz= z^{-1/2}e^{z^2/2}\mathcal{W}\left(\frac{1}{4},\frac{1}{4},z^2\right)$$
is Liouvillian, is not elementary, but still is rigid. This is because this integral admits a representation in terms of the Whittaker function $\mathcal{W}$, which is a rigid function (as a confluence of the Gauss hypergeometric equation).\\

\noindent
The elliptic integrals of the first kind
$$\int \frac{1}{\sqrt{(z-a)(z-b)(z-c)}} dz$$
These can be expressed as Heun function \cite{10}. The cross ratio (see \cite{11}) is not a constant with respect to $a,b,c$. So if this integral could be expressed using a hypergeometric function (Whittaker functions are excluded due to the irregular singularity), then it would lead to a parametric algebraic transformation between Heun and hypergeometric function. Such parametric transformations have been classified in \cite{7}, and a transformation for general elliptic integral of the first kind is not included. This integral cannot be expressed in terms of elementary functions in general, and thus is not rigid for generic $a,b,c$. Remark that some exceptional values of $a,b,c$ for which this function is rigid are known
$$\int \frac{1}{\sqrt{z^3-1}} dz=-\frac{2}{\sqrt{z}} \,_2F_1\left(1/2,1/6,7/6,\frac{1}{z^3} \right)$$
$$\int \frac{1}{\sqrt{z^3-z}} dz=-\frac{2}{\sqrt{z}} \,_2F_1\left(1/2,1/4,5/4,\frac{1}{z^2} \right)$$
$$\int \frac{1}{\sqrt{z^3-a^2z^2}} dz =\frac{2}{a} \hbox{arctan}\left(\frac{\sqrt{z-a^2}}{a}\right)$$
As the monodromy group is these cases can be explicitly computed, this implies that the complex lattice of corresponding elliptic functions can be explicitly computed. And indeed, those have exceptional properties, being triangular, square and collapsed respectively.

In general, it is difficult to prove that some function is not rigid. In the case of a differential equation of order $2$ with generic exponents at singularities, we cannot find an algebraic pullback mapping these singularities to only $3$ points, because it would imply a rational relation between the exponents (see pullbacks of hypergeometric equation to themselves in \cite{13}). The cases with parameters are thus much easier, and it is mostly done for Heun functions in \cite{7}, as exponents typically depend on the parameter. This is exactly our situation as the energy level $E$ is a parameter in our Schroedinger equation, and appears in the asymptotic behaviour of the solutions.

\section{Rigid eigenfunctions}

We now proceed to the proof of Theorem \ref{thmmain0}.\\

\textbf{Notation}: In the following of the article, we will use both the notations $f(z,E)$ and $f(z)$, which will be two different functions, and we will always precise the variables in case of ambiguity.

\subsection{Asymptotic analysis}

The possible asymptotic behaviours of eigenfunctions are of the following form
$$z^\gamma e^{\sum\limits_{k=1}^{\max(0,n)+2} a_k z^{k/2}}$$
with $n$ the asymptotic exponent of $V$ at infinity (if $V$ converges at infinity, we always assume that it converges to $0$). The exponent $\gamma$ encodes the formal monodromy at infinity. Let us look now at the dependence in function of $E$. We inject this expression in equation \eqref{eq1}, giving
\begin{itemize}
\item if $n\geq 3$, then $\gamma$ and all the $a_i$ are constant
\item if $n=2$, then $\gamma$ is affine in $E$ and all the $a_i$ are constant
\item if $n=1$, then $\gamma=-1/4$, $a_3$ constant, $a_2=0$ and $a_1$ linear in $E$.
\item if $n=-1$, then $a_2=\sqrt{-E}$, $\gamma\sqrt{-E}$ constant.
\item if $n\leq -2$, then $a_2=\sqrt{-E}$, $\gamma=0$.
\end{itemize}

To find rigid eigenfunctions, we need to search two type of rigid solutions: the ones coming from hypergeometric or Whittaker functions, and the Liouvillian ones. Let us first remark that equation \eqref{eq1} has an irregular singularity at infinity. If a solution is rigid and non Liouvillian, then its expression has to involve a non solvable hypergeometric or Whittaker function, which are solutions of equations of order $2$. Thus they should have an expression of the form
\begin{equation}\label{eqpsi}
\psi(z,E)=h(z,E)(M(z,E)F(f(z,E))+F'(f(z,E)))
\end{equation}
with $h$ hyperexponential, $f,M$ algebraic and $F$ hypergeometric or Whittaker, with parameters possibly depending on $E$. The quotient of the two possible asymptotic behaviours of solutions of equation \eqref{eq1} has always an essential singularity at infinity. If the function $F$ was hypergeometric, then the quotient of two solutions of the form \eqref{eqpsi} would have a Puiseux/log series at infinity. Thus the function $F$ should be of Whittaker type.\\

So non Liouvillian eigenfunctions should be of the form
$$e^{\int g(z,E) dz}\left(M(z,E)\mathcal{W}(\mu(E),\nu(E),f(z,E))+\mathcal{W}'(\mu(E),\nu(E),f(z,E))\right)$$
where $g,M,f$ are algebraic in $z,E$ ($E$ being the parameter). Computing the second or differential equation whose solution is this function, the $g$ can be expressed in function of $f,M$ through the condition that no term in $\psi'$ appears in this equation. We then obtain for $V+E$ a large rational expression depending on $M,f$, their derivatives in $z$ and $E,\nu(E),\mu(E)$.

\subsection{Whittaker pullbacks}

A rigid solution related to (a non solvable) Whittaker function of a second order differential equation is of the form
$$e^{\int g(z) dz}\left(M(z)\mathcal{W}(\mu,\nu,f(z))+\mathcal{W}'(\mu,\nu,f(z))\right)$$
with $f,M,g$ algebraic. Still if we restrict ourselves to differential equation with rational coefficients, the $f,M,g$ cannot be arbitrary algebraic functions (see \cite{16} for Bessel functions).

\begin{prop}\label{proppull0}
If the function
$$e^{\int g(z) dz}\left(M(z)\mathcal{W}(\mu,\nu,f(z))+\mathcal{W}'(\mu,\nu,f(z))\right)$$
with $f,M,g$ algebraic is solution of a second order unimodular differential equation with $\mathcal{W}$ non solvable, then $f$ is either rational, or the square root of a rational function with $\mu=0$.
\end{prop}

\begin{proof}
We are searching a pullback transformation and gauge transformation which sends an unimodular differential equation with rational coefficients (the Whittaker equation) to an unimodular differential equation with rational coefficients. Both these transformations are algebraic, as the function $\exp(\int g(z) dz)$ can be expressed algebraically in $M,f,z$ and their derivatives. So the functions we are looking are of the form
\begin{equation}\label{eqtest}
M_1(z)\mathcal{W}(\mu,\nu,f(z))+M_2(z)\mathcal{W}'(\mu,\nu,f(z))
\end{equation}
with $M_1,M_2,f$ algebraic. Thus in particular the pullback function $f$ cannot be an arbitrary algebraic function as this function has to satisfy a rational linear differential equation.

We now consider the function $\sigma$ which sends one value of $f$ to the others (recall $f$ is algebraic and so multivalued). This function is the Galois action on the branches of $f$. And thus the action of $\sigma$ on \eqref{eqtest} produces another solution of the differential equation. As we assumed the differential Galois group of the Whittaker function being $SL_2(\mathbb{C})$, this implies a relation of the form
$$\mathcal{W}_1(\mu,\nu,\sigma(z))=S_1(z)\mathcal{W}_1(\mu,\nu,z)+S_2(z)\mathcal{W}_2(\mu,\nu,z)$$
with $S_1,S_2$ algebraic, $\mathcal{W}_1,\mathcal{W}_2$ a basis of solutions of the Whittaker equation. Let us make some precisions about the function $\sigma$. The differential equation has rational coefficients, the singularities of \eqref{eqtest} should not depend on which branch of $f$ we choose. Thus if $f(z)=0,\infty$ for some branch, this should be the same for the other branches (and $0,\infty$ cannot be exchanged as they lead to a different type of singularity, one regular, the other irregular). This implies in particular that $\sigma$ is univalued at $0,\infty$, and their value are $0,\infty$. Let us now act the Galois group of the Whittaker equation on the above relation. We consider a path in $\mathbb{C}^*$ and a corresponding monodromy/Stokes matrix $A$, assumed to be diagonal (possible as the Galois group is $SL_2(\mathbb{C})$) with eigenvalues $\alpha,1/\alpha$, $\alpha$ not root of unity, acting on the basis $\mathcal{W}_1,\mathcal{W}_2$. As $\sigma$ is algebraic, the path is not always closed on its associated Riemann surface. Still, if we take a suitable power of $A$, noted $B$, this will correspond to a closed path of the Riemann surface associated to $\sigma$ ($\sigma$ being algebraic, its monodromy group is finite). When applying $B$ infinitely many times, this gives after taking a limit
$$\mathcal{W}_1(\mu,\nu,\sigma(z))=S_1(z)\mathcal{W}_1(\mu,\nu,z) \hbox{ or}$$
$$\mathcal{W}_1(\mu,\nu,\sigma(z))=S_2(z)\mathcal{W}_2(\mu,\nu,z)$$
Using the unimodular property, we get that $S_{1,2}(z)=c\sqrt{\sigma(z)},\; c\in\mathbb{C}$.

Let us now consider a ramification point $\alpha$ of $\sigma$, outside $0,\infty$. We have the relation
$$\frac{\mathcal{W}_1(\mu,\nu,\sigma(z))}{\sqrt{\sigma(z)}}=c\mathcal{W}_i(\mu,\nu,z)$$
Near $\alpha$, the righthandside is analytic. The function $\mathcal{W}_1$ can be chosen arbitrary (solution of the Whittaker equation), and is analytic at $\alpha$. Thus the left hand side (for a generic choice of $\mathcal{W}_1$) is not analytic. Thus such ramification point $\alpha$ does not exist.

This implies that $\sigma$ has at most two ramification points $0,\infty$, and moreover knowing that $\sigma(0)=0,\sigma(\infty)=\infty$, this implies that
$$\sigma(z)=az^r,\quad a\in\mathbb{C}^*,\; r\in\mathbb{Q}^+_*$$
Now looking at asymptotics near infinity, we find that the only possible exponent is $r=1$. Now we need to express $\mathcal{W}_1(\mu,\nu,az)$ in function of $\mathcal{W}_1,\mathcal{W}_2$. The only possibility is $a=1$ or $a=-1$ with $\mu=0$. This implies that $f$ is either rational, or the square root of a rational function with $\mu=0$.
\end{proof}\qed

We now need to find $M,f$ leading to a function $V$ depending only on $z$ and not on the parameter $E$.

\begin{defi}
We consider a linear differential equation with a parameter $E$. We say that this equation has no mobile singularity if the position of the singularities does not depend on $E$. Similarly, we say that a rational function has no mobile singularity (respectively root) if its poles (respectively root) do not depend on the parameter $E$.
\end{defi} 

The Schroedinger equation has no mobile singularities. So the solutions \eqref{eqpsi} should have no mobile singularities. In particular, the Whittaker functions have ramification points at $0,\infty$, and this will restrict the possible pullback transformations $f$.

\begin{prop}\label{proppull}
The pullback function $f(z,E)$ has to be of the form
$$w(E)f(z) \hbox{ or }$$
$$(w_1(E)f(z)+w_2(E))^k \hbox{ with } \nu=\pm\frac{1}{2k},\; k\in\mathbb{N}\setminus \{0,1\} \hbox{ or } $$
$$(w_1(E)f(z)+w_2(E))^k \hbox{ with } \nu=\pm\frac{1}{2k},\; \mu=0,\; k\in\frac{1}{2} \mathbb{N} \setminus \{0,1/2,1\}$$
\end{prop}

\begin{proof}
Let us first consider the case $\mu\neq 0$. Then we have $f(z,E)$ rational in $z$ according to Proposition \ref{proppull0}. The Schroedinger equation has no mobile singularities. The values $f=\infty$ always lead to singularities of \eqref{eqtest}, and so cannot depend on $E$. The point $0$ is a regular singularity of the Whittaker equation with exponents $1/2+\nu,1/2-\nu$. Let us consider a root $\alpha$ of $f$, with multiplicity $k\in\mathbb{N}^*$. The function $\mathcal{W}(\mu(E),\nu(E),f(z,E))$ admits a Puiseux series in $z$ near $z=\alpha$, with first term exponent $(1/2+\nu)k$ or $(1/2-\nu)k$. Now if $\alpha$ depends on $E$, it cannot be a singularity of the Schroedinger equation (not even an apparent one). Now taking into account the gauge transformation, the function $\psi$ admits a Puiseux series with first term exponent $(1/2+\nu)k+\gamma$ or $(1/2-\nu)k+\gamma$, with $\gamma$ depending on the gauge transformation. If we want $\alpha$ not being a singularity of the Schroedinger equation, we need these exponents to be $0,1$. And thus we need
$$(1/2+\nu)k-(1/2-\nu)k=\pm 1$$
and thus $\nu=\pm 1/(2k)$. Remark moreover that $\nu=\pm 1/2$ leads to a logarithmic singularity for the Whittaker function, and thus $\alpha$ would always be a singularity.

Thus if $\nu\notin \{\pm 1/(2k), k\in\mathbb{N}\setminus \{0,1\}\}$, then all the roots of $f$ do not depend on $E$. And so is of the form $w(E)f(z)$.
Let us now assume $\nu=\pm 1/(2k), k\in\mathbb{N}\setminus \{0,1\}$. The pullback function is of the form $F(z)Q(z,E)^k$, with $F$ rational, $Q$ polynomial with simple roots in $z$. Let us now look at critical points of $f$. If such a critical point $\alpha$ is not on the level $f=0$, then it will give a singularity. We have
$$f'(z,E)=Q(z,E)^{k-1}(F'(z)Q(z,E)+F(z)Q'(z,E))$$
and so the right factor $F'(z)Q(z,E)+F(z)Q'(z,E)$ cannot have roots depending on $E$ (as else it would lead to a mobile singularity). And so
$$F'(z)Q(z,E)+F(z)Q'(z,E)=w_1(E)S(z)$$
This is a non homogeneous linear differential equation in $Q$, and the solutions are of the form
$$Q(z,E)=w_2(E)F(z)^{-1/k}+w_1(E)P_1(z)$$
Let us remark that we can assume that $w_1,w_2$ are not $\mathbb{C}$-dependant, as this would lead again to a pullback function of the form $w(E)f(z)$. And so both functions $F(z)^{-1/k},P_1(z)$ have to be polynomials. Let us note $F(z)=1/P_2(z)^k$, giving
$$f(z,E)=\left(w_2(E)+w_1(E)\frac{P_1(z)}{P_2(z)} \right)^k$$
which gives the second case of the Proposition.

Let us now consider the case $\mu=0$. Then $f(z,E)$ is a square root of a function rational in $z$. The same arguments as before still apply, except that the root multiplicity $k$ can be half-integer. So if $\nu\notin \{\pm 1/(2k), k\in1/2\mathbb{N}\setminus \{0,1/2,1\}\}$, then the pullback is of the form $w(E)f(z)$, and else is of the form
$$f(z,E)=\left(w_2(E)+w_1(E)\frac{P_1(z)}{P_2(z)} \right)^k$$
giving the third case of the Proposition.

\end{proof}\qed

\begin{prop}\label{proppul2}
The possible pullback functions $f(z,E)$ for non-Liouvillian eigenfunctions are up to affine transformation
\begin{itemize}
\item $z^2$ with $\mu=E/4$, $\nu$ constant.
\item $2iz\sqrt{E}$ with $\mu=1/(2i\sqrt{E})$, $\nu$ constant.
\item $2iz\sqrt{E}$ with $\mu=0$, $\nu$ constant.
\item $\frac{4i}{3}(z+E)^{3/2}$ with $\mu=0, \nu=\pm 1/3$.
\end{itemize}
\end{prop}

\begin{proof}
Let us note $n$ the asymptotic exponent of $V$ at infinity. Recall that we can always assume that when $V$ converges at infinity, it converges to $0$. And thus that $n\in\mathbb{Z}^*$. The proof is made by disjunction of cases of possible $n\in\mathbb{Z}^*$.

We first remark that if $n\geq 2$, then the pullback function $f$ has to be of the form $w(E)f(z)$ according to the asymptotic expansion, and moreover, $w(E)$ is constant. In this more special case $n\geq 3$, we have moreover that the formal monodromy at infinity is constant. The formal monodromy at infinity of the Whittaker function is encoded by $\mu$, and thus $\mu$ has to be constant. So $f,\mu,\nu$ do not depend on $E$, only $M$ can depend on $E$ (and $g$ as a consequence)
$$\psi(z,E)=e^{\int g(z,E) dz}\left(M(z,E)\mathcal{W}(\mu,\nu,f(z))+\mathcal{W}'(\mu,\nu,f(z))\right)$$
After computation, we find that $g$ is algebraic in $f,E$ and their derivatives. So we make a series expansion of $\psi$ in $E$ near $E=\infty$. After multiplying by a suitable power of $E$, this produces a limit function $s(z)$, smooth almost everywhere. And thus $\psi''(z,E)/\psi(z,E)$ has a limit when $E$ tends to infinity. Impossible as $\psi''(z,E)/\psi(z,E)=V(z)+E$.

Let us now study the case $n=2$. Recall that $\mu$ is affine (non constant) in $E$, so we can make a parameter change and consider $\mu$ as the parameter instead of $E$. So $\psi''(z,E)/\psi(z,E)$ should be affine in $\mu$. Making a series expansion of this at $\mu=\infty$, we find
$$\mu\frac{f'(z)^2}{f(z)}+o(\mu)$$
Thus after possibly affine variable change, we can assume $f(z)=z^2$. This gives by the way a relation between $\mu$ and $E$, $E=4\mu+C$, and we can assume the constant $C$ is zero by putting it into $V$. To conclude, remark that $\nu$ is the exponent at the singularity $0$ of the Whittaker function. In the Schroedinger equation, the exponents do not depend on $E$, and thus so does $\nu$.\\

For the case $n=1$, the asymptotic should be of the form
$$z^{-1/4} e^{a_3z^{3/2}+a_1(E)z^{1/2}}$$
with $a_1$ affine in $E$. Looking at the possible pullbacks in Proposition \ref{proppull}, the only possible one is
$$w_1(E)(f(z)+w_2(E))^{3/2},\;\; \mu=0,\;\nu=\pm 1/3$$
and thus $w_1$ is constant, $w_2$ affine in $E$. After possibly adding a constant to $V$, we can assume $w_2(E)=E$. We find that $g$ is algebraic in $f,E$ and their derivatives and we express $\psi''(z,E)/\psi(z,E)$ in function of $f,M,E$ algebraically. We then make a series expansion at $E=\infty$, giving
$$-\frac{9}{16}w_1^2f'(z)^2E+o(E)$$
So after possibly affine variable change, we can assume $f(z)=z$ and thus $w_1=4i/3$.\\

For $n=-1$ the pullback function $f$ is of the form $w(E)f(z)$, and combining this with the asymptotic expansion, we have that $f(z,E)$ is of the form $f(z)\sqrt{-E}$. 

The asymptotic data also give us that $\gamma\sqrt{-E}$ is constant, where $\gamma$ encodes the formal monodromy exponent at infinity of $\psi(z,E)$. The $\mu$ parameter in $\mathcal{W}$ encodes the formal monodromy of $\mathcal{W}$ at infinity, and thus we obtain that 
$\gamma=\mu$ up to a (integer) constant. Thus we have a relation of the form (knowing that $\gamma\neq 0$)
$$E=\frac{\alpha}{\mu^2}+\beta,\qquad \alpha\neq 0$$
Remark that we can assume $\beta=0$ as a constant can be put in the potential $V$. We can also assume $\alpha=-1/4$ by making a dilatation of the coordinate system (which multiplies $E$ by a constant), giving $\mu=1/(2i\sqrt{E})$. The parameter $\nu$ also cannot depend on $E$ as the exponents of the Schroedinger equations do not depend on $E$. We obtain a large expression for the potential $V$ depending on $M,f,E$, and we make a series expansion at $E=\infty$, giving
$$-\frac{1}{4}f'(z)^2E+O(1)$$
Thus we have that $f(z)=2iz$ (up to affine variable change), and thus the the pullback function is $2iz\sqrt{E}$.

For $n\leq -2$, according to asymptotics, we need to have $\mu=0$. Using Proposition \ref{proppull}, the possible pullbacks are of the form $w(E)f(z)$ or
$$w_1(E)(f(z)+w_2(E))^k,\; k\in\frac{1}{2} \mathbb{Z} \setminus \{0,1/2,1\},\; \nu=\frac{1}{4k}$$
Now using the asymptotics in $z$, we have
$$w_1(E)(f(z)+w_2(E))^k=z\sqrt{-E}+O(1)$$
For $k\geq 3/2$, such series expansion is impossible, and so the only possible pullbacks are of the form $w(E)f(z)$.

With the asymptotics, we obtain moreover $w(E)=\sqrt{-E}$. Computing the corresponding potential and making a series expansion at $E=\infty$, we obtain
$$-\frac{1}{4}f'(z)^2E+o(E)$$
And thus up to affine variable change the pullback is of the form $2iz\sqrt{E}$.

\end{proof}\qed

\subsection{Liouvillian pullbacks}

Let us now look at Liouvillian eigenfunctions. The finite projective case and the log case of Theorem \ref{thmrigid} are impossible due to the asymptotic behaviour at infinity. For the diagonal case, we need to search all Schroedinger equation with diagonal Galois group (for all $E$). This implies that there exists a hyperexponentional solution for all $E$.

\begin{lem}\label{lemhyper}
If the Schroedinger equation \eqref{eq1} has one hyperexponential solution
$$\psi(z,E)=e^{\int F(z,E) dz}$$
then the space of solutions of equation \eqref{eq1} is of the form
\begin{equation}\label{eqhyper}
e^{\int g(z,\sqrt{-E})dz}\left(A\frac{M(z,E)}{\sqrt{-E}} \hbox{ch}(z\sqrt{-E})+B\hbox{sh}(z\sqrt{-E})\right) \qquad A,B\in\mathbb{C}
\end{equation}
with $g,M$ rational in both variables.
\end{lem}

\begin{proof}
We first write $F$ under partial fraction decomposition. After integration, we obtain a logarithmic part and a rational part. We know there are no mobile singularities, that the residues are constant, and that in the finite irregular singularity case the exponential part does not depend on $E$, we deduce that
$$e^{\int F(z,E) dz}=\prod\limits_{i=1}^p (z-z_i(E))^{\alpha_i} e^{P(z,E)+H(z)}$$
with $P$ polynomial in $z,E$ and $H$ rational in $z$ of negative degree. We now look at possible asymptotic expansions at $z=\infty$. Let us look at the monodromy at infinity. Recall that for $n=2,-1$, the true monodromy should depend on $E$. However, all the $\alpha_i$ are constant in $E$. The Stokes phenomenon here is trivial, and thus the true monodromy at infinity is the sum of the $\alpha_i$. So it cannot depend on $E$. For $n=1$, non rational terms are required, and so is also impossible.

For $n\geq 3$, we obtain that $P$ is constant in $E$ when $n\geq 1$, and true monodromy at infinity is constant. This implies
$$\lim\limits_{E\rightarrow\infty}  e^{\int F(z,E) dz} E^{\beta}=s(z)\neq 0$$
for a suitable $\beta$. Thus $\psi(z,E)$ would converge after rescaling to $s(z)$. Impossible as $\psi''(z,E)/\psi(z,E)=V(z)+E$. Thus $n\leq -2$.

We conclude that $P(z,E)=z\sqrt{-E}$. We can now act the Galois group element $\sqrt{-E} \longrightarrow -\sqrt{-E}$ to obtain for free a new hyperexponential solution. Let us now remark that if an $\alpha_i$ is not a positive integer, then $z_i(E)$ is a singularity of the Schrodinger equation. And thus $z_i$ cannot depend on $E$. Such a term can be put in factor for both hyperexponential solutions. This gives a solution space of the form
$$e^{\int g(z,\sqrt{-E})dz}\left(Ae^{z\sqrt{-E}} T(z,\sqrt{-E})+Be^{-z\sqrt{-E}} T(z,-\sqrt{-E})\right) \qquad A,B\in\mathbb{C}$$
with $T$ polynomial, $g$ rational. Rewriting the terms using $\hbox{ch},\hbox{sh}$ instead of exponentials, and then putting in the coefficient in front of $\hbox{sh}$ and then inside the $e^{\int g(z,\sqrt{-E})dz}$ by changing $g$, we obtain the expression \eqref{eqhyper}. And the expression of $M$ in function of $T$ proves that it is indeed rational in $E$.

\end{proof}\qed

This Lemma implies in particular that if we have one hyperexponential solution, then we have two and the solution space of equation \eqref{eq1} is the solution space of a differential equation which is a rational gauge transformation of the equation $y''+Ey=0$.\\

Let us now look at the dihedral case, for which the space of solutions is of the form
$$Ae^{\int R_1(z) dz}+Be^{\int R_2(z)dz}$$
with $R_1,R_2$ belonging to an extension of degree $2$ over $\mathbb{C}(z)$. Remark all rigid functions in the dihedral case are not algebraic pull-backs of Galois rigid operators with dihedral Galois group. Indeed, we need to take into account the field operations, as in the following example
$$(\sqrt{z^2+1}+1)^{\sqrt{2}} e^{z\sqrt{z^2+1}}$$
This is not a gauge transformation/algebraic pullback of a solution of a dihedral Galois rigid equation. However, the term in the exponential can be written $\sqrt{z^2(z^2+1)}$ and so is a pullback of $\exp{\sqrt{z}}$. So is $(\sqrt{z^2+1}+1)^{\sqrt{2}}$ with $\sqrt{z^2+1}$. Still the expression in the exponential has to be elementary, and thus the integrals $\int R_i(z) dz$ have to be elementary.

\begin{prop}
A Schroedinger equation \eqref{eq1} cannot have a rigid solution space with dihedral Galois group.
\end{prop}

\begin{proof}

The space of solutions is of the form
$$Ae^{\int R_1(z,E) dz}+Be^{\int R_2(z,E)dz}$$
where $R_1,R_2$ are solutions of the Ricatti equation associated to \eqref{eq1}. This equation has base field coefficients $\mathbb{C}(E)$. We now use Theorem 5.4 of \cite{21}, saying that $R_1,R_2$ are in an extension of degree $2$ over $\mathbb{C}(z,E)$. And thus $R_1,R_2 \in \mathbb{C}(z,E,\sqrt{f(z,E)})$ for some polynomial $f$. We can now rewrite the solution space under the form
$$e^{\int g(z,E) dz} \left(Ae^{\int \sqrt{f(z,E)}F(z,E) dz}+Be^{-\int \sqrt{f(z,E)}F(z,E) dz} \right) $$
with $g,f,F$ rational functions in $z,E$. We also know that the Schroedinger equation is unimodular, giving a condition on $g$ allowing it to be expressed in function of $f,F$. We find in particular that the solution space should be of the form
$$\frac{1}{f(z,E)^{1/4} \sqrt{F(z,E)}}\left(Ae^{\int \sqrt{f(z,E)}F(z,E) dz}+Be^{-\int \sqrt{f(z,E)}F(z,E) dz} \right) $$

Let us now look at singularities. The solutions should not have mobile singularities, and thus $f(z,E),F(z,E)$ cannot have mobile roots/poles. And thus we can write our solution space under the form
$$\frac{1}{f(z)^{1/4} \sqrt{F(z)}}\left(Ae^{\int \sqrt{w(E)f(z)}F(z) dz}+Be^{-\int \sqrt{w(E)f(z)}F(z) dz} \right)$$
with $f,F$ rational in $z$, $w$ rational in $E$. We can moreover assume $f(z)$ polynomial with only simple roots.

Now recall that our solution has to be rigid as well. And thus the integral has to be an elementary function. So we have
$$\int \sqrt{w(E)f(z)}F(z) dz=\sqrt{w(E)f(z)}Q(z)+\sum\limits_i \alpha_i\sqrt{w(E)} \ln  u_i(z)$$
with $Q\in\mathbb{C}(z),u_i\in\mathbb{C}(z,\sqrt{f(z)})$.

Let us remark that if $w$ is constant, then we would obtain a solution space not depending in $E$, which is impossible. On the other hand, we know that in the Schroedinger equation \eqref{eq1}, there are no mobile singularities, the residues are constant, and in the finite irregular singularity case the exponential part does not depend on $E$. So we deduce that
$Q$ is a polynomial.

We now look at the asymptotic behaviour in $z$. The cases $n\geq 2$ are impossible, as $E$ has to appear in the exponential part. The case $n=1$ is also impossible as $\sqrt{w(E)}$ has to appear as a factor of all terms in the exponential. Remain the cases $n\leq -1$, for which we should have
$$\sqrt{w(E)f(z)}Q(z) \sim \sqrt{-E}z$$
As $f$ is a polynomial and cannot be constant (as else the Galois group would be diagonal instead of dihedral), we conclude that $f$ is of degree $2$, and $Q$ is constant. And so after possibly an affine coordinate change, we obtain
$$w(E)=-E,\quad f(z)=z^2+1,\quad Q(z)=1$$
Thus the Schroedinger equation should have a solution of the form
$$\psi(z,E)= \frac{ S(z)^{\sqrt{-E}} e^{\sqrt{-E}\sqrt{z^2+1}}}{(z^2+1)^{1/4} \sqrt{\frac{S'(z)}{S(z)}+\frac{z}{\sqrt{z^2+1}} }} $$
with
$$S(z)=\prod\limits_i u_i(z)^{\alpha_i}$$
We then compute a series expansion of $-\psi/\psi$ at $E=\infty$, giving
$$-\frac{\psi''(z,E)}{\psi(z,E)}=E\left(\frac{S'(z)^2}{S(z)^2}+\frac{2zS'(z)}{\sqrt{z^2+1}S(z)}+\frac{z^2}{z^2+1} \right) +o(E)$$
The integrability condition is
$$\frac{S'(z)^2}{S(z)^2}+\frac{2zS'(z)}{\sqrt{z^2+1}S(z)}+\frac{z^2}{z^2+1}=1$$
and this cannot be satisfied for a function $S(z)$ of the required form.

\end{proof}\qed

\subsection{Proof of Theorem \ref{thmmain0}}

The hyperexponential factor $\exp{\int g(z,\sqrt{-E})dz}$ can be computed explicitly by just requiring that the differential equation whose solution is the eigenfunction does not have a $\psi'$ term. This gives a first order differential equation, and this equation admits surprisingly very simple solutions, giving the prefactors in the expressions of Theorem \ref{thmmain0}. The Liouvillian case of Lemma \ref{lemhyper} is included as a special case of the third case of Theorem \ref{thmmain0}, with $\nu=1/2$. Indeed, for $\mu=0,\nu=1/2$, the Whittaker function simplifies and simply becomes an exponential. Finally we have to prove that $M(z,E)$ is not only algebraic but rational. In the Liouvillian case, this is already proved by Lemma \ref{lemhyper}. We remark that $M$ as written in Theorem \ref{thmmain0} corresponds to a gauge transformation of a differential equation with rational coefficients in $z,E$. In the non solvable case, if $M$ was algebraic and not rational, the Galois action would generate more solutions. Impossible, as we already have a vector space of dimension $2$ of solutions (a single solution of a non solvable equation generates a basis of solutions under the action of the differential Galois group).

Remark that the expressions of the gauge transformations in Theorem \ref{thmmain0} are not chosen as simple as possible, but these choices will simplify proofs of Theorem \ref{thmmain}.

\section{Classification of integrable potentials}

We will now describe all the possible $M$ of Theorem \ref{thmmain0} leading to a potential $V$.

\subsection{Mobile singularities}

We first prove the following Proposition, valid for the $4$ families of eigenfunctions of Theorem \ref{thmmain0}

\begin{prop}\label{propsing}
Let $\psi$ be a function of the form given by Theorem \ref{thmmain0} and $H$ the denominator under the square root. The function $\psi$ is a solution of a Schroedinger equation if and only if
$$H(z,E)=\frac{w(E)P(z)}{Q(z,E)}$$
with $w,P,Q$ polynomials and $\deg_E w\geq \deg_E \hbox{numer}(M)+\deg_E \hbox{denom}(M)+1$.
\end{prop}

\begin{proof}
Let us first remark that as the Schroedinger equation is linear, the singularities of the solutions are always singularities of the potential. And thus the denominators in the expression of Theorem \ref{thmmain0} are singularities of the equation, and thus poles of $V$. But then these poles should not depend on $E$. Then we can write
$$H(z,E)=\frac{w(E)P(z)}{Q(z,E)}$$
with $w,P,Q$ polynomials.

We have that $\deg_E \hbox{numer}(H)=\deg_E w$ and so we just have to prove that $\deg_E \hbox{numer}(H)\geq \deg_E \hbox{numer}(M)+\deg_E \hbox{denom}(M)+1$. Just noting $M=R/S$, replacing and taking the numerator, we find $\deg_E w=\max(2\deg_E R,2\deg_E S+1)$. This however skips the possible problem of simplifications of the fraction.

Assume there is a simplification. It would mean that a pole $\alpha(E)$ (root of $S$) is not a pole of $H$. Looking at the expression of $\psi$, this implies that $\alpha(E)$ would be a singularity of $\psi$ (as no simplification can occur with the denominator $\sqrt{H}$). And thus $\alpha(E)$ does not depend on $E$, as mobile singularities are forbidden. So factors depending on $z$ only could simplify, but this does not change the degrees in $E$. And thus we have always
$$\deg_E w\geq \max(2\deg_E R,2\deg_E S+1)\geq \deg_E \hbox{numer}(M)+\deg_E \hbox{denom}(M)+1$$

Now we prove the other way of the Proposition. We express $V(z)$ as a rational expression in $E,M$ and its derivatives. We then replace the derivatives of $M$ using $H$. We find that this expression has no singularities at the roots of $w$ (which are roots of $H$ and its derivatives). Knowing that $\deg_E w\geq \deg_E \hbox{numer}(M)+\deg_E \hbox{denom}(M)+1$, we obtain that the degree of $V$ is $0$. And thus $V$ only depend on $z$, and so is a potential.
\end{proof}\qed

Using Proposition \ref{propsing}, we have that $M$ is completely determined using Hermite rational interpolation by its evaluations at roots of $w$, and possibly derivatives in $E$ for multiple roots of $w$. Moreover, such a $M$ will always lead to a quantum integrable potential. Let us now look at what happen at a root $E_0$ of $w$. The denominator $H$ should vanish at $E_0$. We now pose for the following of the proof
$$M(z,E)=-\frac{Y'(z,E)}{Y(z,E)}$$
The equation $H=0$ becomes respectively in each of the $4$ cases of Theorem \ref{thmmain0}
\begin{equation}\begin{split}\label{eqsing}
z^2Y''(z,E_0)-zY'(z,E_0)-(z^4-E_0z^2-4\nu^2+1)Y(z,E_0)=0\\
4z^2Y''(z,E_0)+(4E_0z^2+4z-4\nu^2+1)Y(z,E_0)=0\\
4z^2Y''(z,E_0)+(4E_0z^2-4\nu^2+1)Y(z,E_0)=0\\
Y''(z,E_0)+(z+E_0)Y(z,E_0)=0
\end{split}\end{equation}
Now a simple necessary condition for getting a quantum integrable potential is that for all roots of $w$, these equations admit a hyperexponential solution (and this is sufficient when $w$ has no multiple root). The following of the proof of Theorem \ref{thmmain} will be split in $4$ parts, each corresponding to one possible eigenfunction class given by Theorem \ref{thmmain0}.

\subsection{Case $1$}

\begin{lem}[Galois groups of Whittaker equation in \cite{17}]\label{lem1}
The hyperexponential solutions of the first equation \eqref{eqsing} are
$$z^{2\epsilon_1\epsilon_2\nu+1}e^{-\epsilon_1z^2/2} \!_1F_1(-k,2\epsilon_1\epsilon_2\nu+1,\epsilon_1z^2)$$
with $E_0=\epsilon_1 (4k+2)+4\epsilon_2\nu,\; k\in\mathbb{N},\; \epsilon_1,\epsilon_2=\pm 1$.
\end{lem}

Remark that it is possible that the equation admits $2$ hyperexponential solutions for some specific $E_0$. This case corresponds to when $E_0$ can be written $\epsilon_1 (4k+2)+4\epsilon_2\nu$ in two different ways. Now this Lemma gives a condition on the roots of $w$, and gives a formula for $M$ at the roots of $w$. If $w$ has no multiple roots, $M$ can be recovered through Pade interpolation, giving Theorem \ref{thmmain} in the case $1$.

We now focus on multiple roots of $w$. As $w$ vanishes at some $E_0$ at order $p\geq 2$, we can differentiate equation \eqref{eqsing} with respect to $E_0$, giving us additional equations. Now the condition for getting a quantum integrable potential is that the logarithmic derivative in $z$ of the series solution $Y$ in $E$ at order $p$ has rational coefficients in $z$ (the function $M$ can then be recovered by Pade Hermite interpolation).

\begin{lem}\label{lem2}
Assume $w$ has a double root at $E_0=\epsilon_1 (4k+2)+4\epsilon_2\nu$ with $Y(z,E_0)$ the hyperexponential function given by Lemma \ref{lem1}. The function $Y$ leads to a quantum integrable potential if and only if $2\epsilon_1\epsilon_2\nu+k\in \mathbb{N}$.
\end{lem}

\begin{proof}
We have the equation
$$M(z,E)^2z^2+M(z,E)z-M'(z,E)z^2-z^4+z^2E-4\nu^2+1=O(E^2)$$
Using Lemma \ref{lem1}, we know that
$$M(z,E_0)=-\frac{Y'(z)}{Y(z)}$$
with
$$Y(z)=z^{2\epsilon_1\epsilon_2\nu+1}e^{-\epsilon_1z^2/2} \!_1F_1(-k,2\epsilon_1\epsilon_2\nu+1,\epsilon_1z^2)$$
Noting
$$M(z,E)=-\frac{Y'(z)}{Y(z)}+(E-E_0)M_1(z)+O((E-E_0)^2)$$
and injecting it in the equation of $M$ above, we find the solutions for $M_1$
$$M_1(z)=\frac{z}{Y(z)^2}\int \frac{Y(z)^2}{z} dz$$
The condition for leading to a quantum integrable potential is that $M(z,E)$ should be rational in $z,E$, and thus that $M_1(z)$ should be rational in $z$. Looking at the integral above, this condition becomes
$$\int z^{4\epsilon_1\epsilon_2\nu+1}e^{-\epsilon_1z^2}\!_1F_1(-k,2\epsilon_1\epsilon_2\nu+1,\epsilon_1z^2)^2 dz \in  e^{-\epsilon_1z^2}z^{4\epsilon_1\epsilon_2\nu}\mathbb{C}(z)$$

We perform a variable change transforming the antiderivative computation in 
$$\int z^{2\epsilon_1\epsilon_2\nu}e^{-\epsilon_1z}\!_1F_1(-k,2\epsilon_1\epsilon_2\nu+1,\epsilon_1z)^2 dz$$
This has to be an element of $e^{-\epsilon_1z}z^{2\epsilon_1\epsilon_2\nu}\mathbb{C}(z)$.

Let us first look at when $\nu\in 1/2\mathbb{Z}$, we are integrating a rational function times exponential. The only possible pole is at $z=0$ (the $\!_1F_1$ is a polynomial). However, for such $\nu$, the function $\!_1F_1$ can become singular. This can be solve using a regularization process, multiplying by some function in $\nu$, giving
$$\frac{\Gamma(2\epsilon_1\epsilon_2\nu+1+k)}{\Gamma(2\epsilon_1\epsilon_2\nu+1)}\;\!\!_1F_1(-k,2\epsilon_1\epsilon_2\nu+1,\epsilon_1z)$$
instead of the $\!_1F_1$. Now the valuation at $z=0$ of
$$z^{2\epsilon_1\epsilon_2\nu}\frac{\Gamma(2\epsilon_1\epsilon_2\nu+1+k)^2}{\Gamma(2\epsilon_1\epsilon_2\nu+1)^2}\;\!\!_1F_1(-k,2\epsilon_1\epsilon_2\nu+1,\epsilon_1z)^2$$
is $2\epsilon_1\epsilon_2\nu+k$, and thus if this quantity is non-negative, we are integrating a polynomial times exponential. The condition is then fulfilled.

Let us now prove that for $2\nu\notin\mathbb{Z}$, the condition cannot be satisfied. Let us note $v_{n,k}$ the coefficients in $z^n$ of the polynomial $\!_1F_1(-k,2\epsilon_1\epsilon_2\nu+1,\epsilon_1z)^2$. Trying to express the antiderivative as an element of $e^{-\epsilon_1z}z^{2\epsilon_1\epsilon_2\nu}\mathbb{C}(z)$, we obtain a big linear system, and when $2\nu\epsilon_1\epsilon_2 \notin\mathbb{Z}$, we can solve it under the condition
$$\sum\limits_{n=0}^{2k}v_{n,k}\epsilon_1^n\frac{\Gamma(2\epsilon_1\epsilon_2\nu+n+1)}{\Gamma(2\epsilon_1\epsilon_2\nu)}=0$$

The coefficients $v_{n,k}$ satisfy holonomic system (see \cite{18,19} for basic properties) with shifts in $n,k$ as a convolution of $P$-finite sequences, the coefficients of the $\!_1F_1$. So is $\epsilon_1^n\Gamma(2\epsilon_1\epsilon_2\nu+n+1)$. As the holonomic property is stable by definite summation, this above sum as a sequence in $k$ also satisfy a recurrence equation. This can be found thanks to the holonomic package \cite{20}, giving the relation
$$\sum\limits_{n=0}^{2k}v_{n,k}\epsilon_1^n\frac{\Gamma(2\epsilon_1\epsilon_2\nu+n+1)}{\Gamma(2\epsilon_1\epsilon_2\nu)}=
\epsilon_1\epsilon_2\frac{2\Gamma(2\epsilon_1\epsilon_2\nu+1)\Gamma(k+1)\nu}{\Gamma(k+1+2\epsilon_1\epsilon_2\nu)}$$
The only admissible root of the righthandside is $\nu=0$, which is excluded. Remark that the poles of the righthandside are when $2\nu\in\mathbb{Z}$, and thus are excluded (these are exactly the singularities of the $\!_1F_1$ function in $\nu$).
\end{proof}\qed

Double roots are only possible for $2\epsilon_1\epsilon_2\nu+k\in \mathbb{N}$. Looking now at the case with simple roots (with generic $\nu$), we can produce a double root by taking a particular $\nu$: indeed, if there are two roots of the form $4\nu +(4k_1+2)$, $-4\nu +(4k_2+2)$, they fuse together when $\nu=(k_2-k_1)/2$. So the double root case at $\epsilon_2 \nu_0 +\epsilon_1 (k_0+1/2)$ can be obtained from the simple root case with two roots
$$\epsilon_2 \nu+ \epsilon_1 (k_0+1/2), -\epsilon_2 \nu +\epsilon_1 (k_0+2\epsilon_2\epsilon_1\nu_0+1/2) \hbox{ if } \epsilon_1\epsilon_2=1$$
$$\epsilon_2 \nu+ \epsilon_1 (k_0-2\epsilon_2\epsilon_1\nu_0+1/2), -\epsilon_2 \nu +\epsilon_1 (k_0+1/2) \hbox{ if } \epsilon_1\epsilon_2=-1 $$
So the double root case is included in the simple root case as a limit for a specific $\nu$.

To conclude, let us prove that triple root or more are not possible

\begin{lem}\label{lem3}
If $w$ has a triple or more root $E_0$, then there is no $Y$ leading to a quantum integrable potential
\end{lem}

\begin{proof}
We have the equation
$$M(z,E)^2z^2+M(z,E)z-M'(z,E)z^2-z^4+z^2E-4\nu^2+1=O(E^3)$$
Using Lemma \ref{lem2}, we can assume $E_0=\epsilon_1(4k+2)+2\epsilon_1(p-k)$ and $\nu=\epsilon_1\epsilon_2(p-k)/2$ with $k,p\in\mathbb{N}$. We also know that
$$M(z,E_0)=-\frac{Y'(z)}{Y(z)}$$
with
$$Y(z)=z^{p-k+1}e^{-\epsilon_1z^2/2} \!_1F_1(-k,p-k+1,\epsilon_1z^2)$$
Noting
$$M(z,E)=-\frac{Y'(z)}{Y(z)}+(E-E_0)M_1(z)+(E-E_0)^2M_2(z)+O((E-E_0)^3)$$
and injecting it in the equation of $M$ above, we find the solutions for $M_2$
\begin{equation}\label{eqM2}
M_2(z)=\frac{z}{Y(z)^2}\int \frac{z}{Y(z)^2} \left(\int \frac{Y(z)^2}{z} dz\right)^2 dz
\end{equation}

As the integrability condition of Lemma \ref{lem2} is satisfied, we know that
$$\frac{z}{Y(z)^2} \int \frac{Y(z)^2}{z} dz \in\mathbb{C}(z)$$
Let us now prove that $M_2$ has monodromy around $0$. More precisely, we will prove that the series expansion at $0$ of
\begin{equation}\label{eqmon}
\frac{z}{Y(z)^2} \left(\int \frac{Y(z)^2}{z} dz\right)^2
\end{equation}
has a non zero residue. However, the residue of this expression does not appear to be holonomic (recall that dividing by $Y(z)$ is a priori forbidden). Let us first remark that
$$\tilde{Y}(z)= Y(z)\int \frac{z}{Y(z)^2} dz$$
is in fact the second independent solution of the first equation \eqref{eqsing}, and thus is holonomic. We now rewrite the residue expression using integration by parts
\begin{align*}
\frac{1}{2i\pi}\oint_0 \frac{z}{Y(z)^2} \left(\int \frac{Y(z)^2}{z} dz\right)^2 dz=\\
\frac{1}{2i\pi}\oint_0 \frac{2}{z}Y(z)^2\int \frac{z}{Y(z)^2} dz\int \frac{Y(z)^2}{z} dz dz= \\
\frac{1}{2i\pi} \oint_0 \frac{2}{z}Y(z)\tilde{Y}(z)\int \frac{Y(z)^2}{z} dz dz
\end{align*}
which is now clearly holonomic. Thus we can find an holonomic system in $p,k$ for the monodromy of this expression around $0$. We now use the holonomic package \cite{20}, and we find the simple formula
$$\frac{1}{2i\pi}\oint_0 \frac{z}{Y(z)^2} \left(\int \frac{Y(z)^2}{z} dz\right)^2 dz=\frac{\Gamma(k+1)\Gamma(p+1-k)^2}{4\Gamma(p+1)}$$
Remark that the formula degenerates for $p<k$. This is due to the singular definition of function $Y$, as it has a pole of order one in such case. Noting that the expression is homogeneous of degree $2$ in $Y(z)$, we expect a pole of order $2$ of the righthandside. We then regularize the formula by a limit process
$$\lim_{\alpha=0}\alpha^2\frac{\Gamma(k+1)\Gamma(p+\alpha+1-k)^2}{4\Gamma(p+\alpha+1)}= \frac{\Gamma(k+1)}{\Gamma(k-p)^2\Gamma(p+1)}$$
We now see that these expressions never vanish for $p,k\in\mathbb{N}$, implying that $M_2$ can never be rational.

\end{proof}\qed

\subsection{Case $2$}

The most important remark here is that the second equation of \eqref{eqsing} is the same as he first equation of \eqref{eqsing} after the variable change
$$Y(z,E)=K\left(\frac{1}{4}z^2E,-\frac{4}{E^2}\right)$$
Thus the previous results of last section apply the same, carrying the variable change in $z$ and $E$. The change of parameter $E$ change accordingly the values of $E_0$ for which the equation admits a hyperexponential solution, and the multiple root cases are also the same after this variable change. We thus obtain the same following results

\begin{lem}\label{lem1b}
The hyperexponential solutions of the second equation \eqref{eqsing} are
$$z^{\epsilon\nu+1/2}e^{-\frac{z}{2\epsilon\nu+2k+1}} \!_1F_1\left(-k,2\epsilon\nu+1,\frac{2z}{2\epsilon\nu+2k+1}\right)$$
with $E_0=-1/(2\epsilon\nu+2k+1)^2,\; k\in\mathbb{N},\; \epsilon=\pm 1$.
\end{lem}

It is still possible that the equation admits $2$ hyperexponential solutions for some specific $E_0$. This case corresponds to when $E_0$ can be written $-1/(2\epsilon\nu+2k+1)^2$ in two different ways. Now this Lemma gives a condition on the roots of $w$, and gives a formula for $M$ at the roots of $w$. If $w$ has no multiple roots, $M$ can be recovered through Pade interpolation, giving Theorem \ref{thmmain} in the case $2$.

We now focus on multiple roots of $w$. As $w$ vanishes at some $E_0$ at order $p\geq 2$, we can differentiate equation \eqref{eqsing} with respect to $E_0$, giving us additional equations. Now the condition for getting a quantum integrable potential is that the logarithmic derivative in $z$ of the series solution $Y$ in $E$ at order $p$ has rational coefficients in $z$ (the function $M$ can then be recovered by Pade Hermite interpolation).

\begin{lem}\label{lem2b}
Assume $w$ has a double root at $E_0=-1/(2\epsilon\nu+2k+1)^2$ with $Y(z,E_0)$ the hyperexponential function given by Lemma \ref{lem1b}. If $Y$ leads to a quantum integrable potential, then $2\epsilon\nu+k\in \mathbb{N}$.
\end{lem}

\begin{lem}\label{lem3b}
If $w$ has a triple or more root, then there is no $Y$ leading to a quantum integrable potential
\end{lem}

Again, double roots are only possible for $2\epsilon\nu+k\in \mathbb{N}$. Looking now at the case with simple roots (with generic $\nu$), we can produce a double root by taking a particular $\nu$: indeed, if there are two roots of the form $-(2\nu+2k_1+1)^{-2}$, $-(-2\nu+2k_2+1)^{-2}$, they fuse together when $\nu=(k_2-k_1)/2$. So as in previous section, the double root case can be obtained from the simple root case by fusing two roots through a limiting process for a specific $\nu$.

\subsection{Case $3$}

The third equation \eqref{eqsing} is a Bessel equation. It has hyperexponential solutions if and only if $E_0=0$, \cite{17}. Thus we have $w(E)=E^k$. The equation for $M$ is then
\begin{equation}\label{eqricatti3}
4M(z,E)^2z^2+4z^2E-4M'(z,E)z^2-4\nu^2+1=O(E^k)
\end{equation}

\noindent\textbf{Case $\nu\notin 1/2\mathbb{Z}$}

At $E=0$, we find only two possible rational solutions
$$M(z,0)= \frac{2\epsilon\nu-1}{2z},\qquad \epsilon=\pm 1$$
Let us write
$$M(z,E)=\sum\limits_{i=0}^{k-1} M_i(z) E^i,\qquad M_0(z)=\epsilon\nu/z$$
Injecting this in the differential equation, we obtain a list of differential equations in the $M_i$ of the form
$$\frac{ 2\epsilon\nu-1}{z} M_i(z)-M'_i(z)=\hbox{Polynomial}(M_j(z)_{j<i})$$
defining the $M_i$ recursively. The important point is that these differential equations are linear, and that the homogeneous part
$$\frac{ 2\epsilon\nu-1}{z} M_i(z)-M'_i(z)=0$$
has no non-zero rational solutions. Thus the equation \eqref{eqricatti3} admits at most two rational solutions (one for each $\epsilon$). In particular, the $M_i$ are uniquely determined by $M_0$.

We now solve equation \eqref{eqricatti3} with zero righthandside. We find two interesting solutions ($\epsilon=\pm 1$)
\begin{equation}\label{eqbessel}
M(z,E)=-\frac{\partial}{\partial z} \ln\left( \mathcal{W}(0,\epsilon\nu,\sqrt{-4E}z)\right)
\end{equation}
These solutions satisfy $M(z,0)=(2\epsilon\nu-1)/(2z)$, and thus their series expansion at order $k$ is the unique series solution we are searching. We now note $M_\epsilon(z,E)$ the rational function obtained by Pade series from the series expansion of \eqref{eqbessel} at order $k$.

The Whittaker function with $\mu=0$ rewrites in terms of the Bessel function. Now recall there is a recurrence relation between the Bessel function. In particular a linear combination of $\mathcal{W}(0,\nu+j,\sqrt{-E}z), \mathcal{W}(0,\nu+1+j,\sqrt{-E}z)$ can be rewritten as a linear combination of $\mathcal{W}(0,\nu,\sqrt{-E}z), \mathcal{W}(0,\nu+1,\sqrt{-E}z)$. And this relation gives a homographic transformation on $M$. Using degree considerations in $E$, we find that the two solutions $M_\epsilon(z,E)$ are such that
$$
\frac{z \left ( \frac{M_\epsilon(z,E)}{\sqrt{-4E}} \mathcal{W}(0,\nu,z\sqrt{-4E})+
\mathcal{W}'(0,\nu,z\sqrt{-4E}) \right)}{\sqrt{4M_\epsilon(z,E)^2z^2+4z^2E-4M_\epsilon'(z,E)z^2-4\nu^2+1}}=
\mathcal{W}(0,\nu+\epsilon k,z\sqrt{-4E})
$$
So these possible $M$ give in fact the same eigenfunction as $M=\infty$, proving third case of Theorem \ref{thmmain}.\\

\noindent\textbf{Case $\nu\in \mathbb{Z}$}

Let us first remark that we can assume $\nu\in [0,1[$ (as we can always shift $\nu$ by an integer). And so we can assume $\nu=0$. So the third equation \eqref{eqsing} becomes
$$4z^2Y''(z,E_0)+(4E_0z^2+1)Y(z,E_0)=0$$
We can differentiate this equation in $E_0$ up to order $k$, the multiplicity of the root $0$ in $w$. Noting
$$Y(z,E)=\sum\limits_{i=0}^{k-1} Y_i(z)E^i+O(E^k),\qquad D=-\partial_z^2-1/(4z^2)$$
we obtain
$$DY_0=0,\;\; DY_{i+1}=Y_i\; i=1\dots k-2$$
So the solutions can be obtained by applying the (pseudo) inverse of $D$ iteratively. Let us remark that the equation
$$Df=g,\quad g\in\sqrt{z}\mathbb{C}[z^2]+\ln(z) \sqrt{z}\mathbb{C}[z^2]$$
has solutions in the vector space $K=\sqrt{z}\mathbb{C}[z^2]+\ln(z)\sqrt{z} \mathbb{C}[z^2]$. So this vector space is stable by these iterations of taking the inverse of $D$. At each step, the degrees of the polynomials in $z^2$ grows by one.

So, possible functions $M$ are given by the series expansion
$$M(z,E)=-\frac{\partial}{\partial z} \ln \left(\sum\limits_{i=0}^{k-1} D^i F(z) E^{k-1-i} \right) +O(E^k)$$
$$F\in K,\hbox{ with degrees } k-1$$
However, we need to check that this series has rational coefficients in $z$. This is not automatic as the function $F$ can contain logs. Knowing that $Y(z,0)$ should have a rational logarithmic derivative, the only possible solutions are $Y(z,0)=a\sqrt{z}$. The constant $a$ can be assumed to be non zero, as we can multiply by a power of $E$ without changing $M$.
This condition rewrites in terms of $F(z)$ by the constraint $F(z)=\sqrt{z}P_1(z^2)+\ln z \sqrt{z} P_2(z^2)$ with $\deg P_1=k-1$. The constant $a$ can further be assumed equal to $1$, after multiplication of $Y$ by a constant, allowing to apply the following Lemma to conclude.

\begin{lem}
A series
$$Y(z,E)=\sum\limits_{i=0}^{k-1} D^i F(z) E^{k-1-i} +O(E^k)$$
with $Y(z,0)=\sqrt{z}$ has a logarithmic derivative which is a series in $E$ with rational coefficients in $z$ if and only if
$$F(z)=\sqrt{z}P_1(z^2)+\ln z \sqrt{z} P_2(z^2),\quad \deg P_2 \leq k/2-1$$
\end{lem}

Let us first assume that the series $Y(z,E)$ can be written under the form
$$Y(z,E)=e^{-\int M(z,E) dz }$$
where $M(z,E)$ is a series in $E$ with rational coefficients. So after integration, we can obtain logs. Putting $E=0$ in the above expression, and knowing that the first term of $Y(z,E)$ is $\sqrt{z}$, we deduce that $M(z,E)=-1/(2z)+O(E)$. The next terms of the series cannot have singularities outside $0$, as $Y$ does not. And the possible singularity at $0$ is of order $1$ at most, due to the form of $Y$ (the singular behaviour at $0$ is in $\ln$). Thus we have
\begin{equation}\label{eqexp}
Y(z,E)=e^{\sum\limits_{i=1}^{k-1} (Q_i(z)+a_i \ln z)E^i +O(E^k)}
\end{equation}
with $Q_i$ polynomials and $a_i$ constants. We now make a series expansion of the righthandside in $E=0$, and we see that powers of $\ln$ can appear. These are impossible as $Y\in K[[E]]$. A necessary condition to avoid powers of logs is that $a_i=0,\forall i=1\dots (k-1)/2$. This implies in particular that the series expansion of $Y$ in $E$ has only polynomials in $z$ as coefficients up to $E^{(k-1)/2}$ included. As the coefficients have the form $D^i F(z)$, we obtain that
$$F(z)=\sqrt{z}P_1(z^2)+\ln z \sqrt{z} P_2(z^2),\quad \deg P_2 \leq k/2-1$$

Now let us prove the opposite way. Assume
$$Y(z)=P_1(z^2)+\ln z P_2(z^2),\quad \deg P_2 \leq k/2-1,$$
and let us prove that $Y$ has a logarithmic derivative which is a series in $E$ with rational coefficients. For $k=1,2$, this can be directly verified. So assume $k\geq 3$. Recall that minus this logarithmic derivative is in fact a solution of the non linear equation
$$4z^2M(z,E)^2+4Ez^2+1-4z^2M'(z,E)=O(E^k)$$
We know that $M$ has a series expansion in $E$ with rational coefficients up to $E^{(k-1)/2}$ included. These are moreover odd functions in $z$. Let us prove that this above equation implies that the next terms of the series are also rational. Noting $M(z,E)=-1/(2z)+\sum_{i=1}^{k-1} M_i(z)E^i$, we obtain the relations from the above equation in $M$
$$\sum\limits_{j=1}^{i-1} M_j(z) M_{i-j}(z)-M_i'(z)-M_i(z)/z=0$$
So these relations give a system of \textbf{linear} differential equations in $M_i,\;i >(k-1)/2$. The solutions are
\begin{equation}\label{eqrecM}
M_i(z)=\frac{1}{z}\int z\sum\limits_{j=1}^{i-1} M_j(z) M_{i-j}(z) dz
\end{equation}
We know that the $M_j$ with $j\leq (k-1)/2$ are odd rational in $z$. Moreover, for $1\leq j\leq (k-1)/2$, they have no singularity at $z=0$, as else the coefficients $a_i$ in equation \eqref{eqexp} would be non zero, and so $Y$ could not satisfy the hypothesis. We also know a priori that $M\in\mathbb{C}(z)[\ln z][[E]]$. So we just have to prove that logs do not appear when making the integration in equation \eqref{eqrecM}.Let us keep track of the valuation at $z=0$ of the $M_j$. We have $\hbox{val} M_j\geq 1,\; 1\leq j\leq (k-1)/2$. Let us prove by recurrence that the $M_i$ have no logs and valuation $\geq -1$. 

For $1\leq i\leq (k-1)/2$, it is already done. For larger $i$, we look the integrand of equation \eqref{eqrecM}, and we see in each product, $M_j$ or $M_{i-j}$ has index  $\leq (k-1)/2$. Thus the valuation of the sum is at least $1-1=0$ (using here the recurrence hypothesis $\hbox{val} M_j\geq -1,\; \forall j<i$). Thus the valuation of the integrand is at least $1$, and so no logs appear in the integration. Moreover, we then divide by $z$, dropping the valuation by $1$, and thus $\hbox{val} M_i\geq -1$. This gives the Lemma, proving fourth case of Theorem \ref{thmmain}.\\

\noindent\textbf{Case $\nu\in 1/2+\mathbb{Z}$}

Let us first remark that we can assume $\nu\in [0,1[$ (as we can always shift $\nu$ by an integer). And so we can assume $\nu=1/2$. So the third equation \eqref{eqsing} becomes
$$Y''(z,E_0)+E_0Y(z,E_0)=0$$
We can differentiate this equation in $E_0$ up to order $k$, the multiplicity of the root $0$ in $w$. Noting
$$Y(z,E)=\sum\limits_{i=0}^{k-1} Y_i(z)E^i+O(E^k),\qquad D=-\partial_z^2$$
we obtain
$$DY_0=0,\;\; DY_{i+1}=Y_i\; i=1\dots k-2$$
So the solutions can be obtained by applying the (pseudo) inverse of $D$ iteratively. Let us remark that the equation
$$Df=g,\quad g\in\mathbb{C}[z]$$
has polynomial solutions. So the vector space of polynomials $\mathbb{C}[z]$ is stable by these iterations of taking the inverse of $D$. At each step, the degrees of the polynomial grows by two.

So, possible functions $M$ are given by the series expansion
$$M(z,E)=-\frac{\partial}{\partial z} \ln \left(\sum\limits_{i=0}^{k-1} D^i F(z) E^{k-1-i} \right) +O(E^k)$$
$$F\in \mathbb{C}[z],\deg F \leq 2k-1$$
Here the series has always coefficients rational in $z$. We finally need to ensure that the precision of the series does not drop by taking the logarithmic derivative, i.e. $Y(z,0)\neq 0$. This implies $\deg F= 2k-1 \hbox{ or } 2k-2$. We then always obtain a rational $M$ through Pade series, proving fifth case of Theorem \ref{thmmain}.

\subsection{Case $4$}

The fourth equation \eqref{eqsing} is an Airy equation, and never has a hyperexponential solution. Thus only the singular $M=\infty$ remains in this case, leading to an affine potential. This proves the sixth case of Theorem \ref{thmmain}.

\section{Examples}

Outside of the special cases $V(z)=z^2+\alpha/z^2,1/z+\alpha/z^2,z,\alpha/z^2$, all the other cases are generated by constructing a gauge transformation function $M$ which is a Pade interpolation or Pade series. In Theorem \ref{thmmain}, these non trivial gauge transformations split in four families, corresponding respectively to eigenfunctions of Theorem \ref{thmmain0} in case $1$, case $2$, case $3$ with $\nu=0$, case $3$ with $\nu=1/2$.

These four families are described completely explicitly: given a set of points or a polynomial (or log-polynomial), we perform a Pade interpolation or Pade series to produce a function $M$, and then a potential $V$. The $4$ families can be generated by algorithms given in the Appendix. The Maple code can be directly copied and is able to generate the integrable potentials of Theorem \ref{thmmain}. The programs take in input a list of elements (for case $1,2$) or a function (for case $3$ with $\nu=0,1/2$).

Here we will make explicit computation of the spectrum for one example of each of the $4$ families. These examples were chosen as they seem to exhibit interesting properties for physical applications.\\

\noindent
\textbf{An anharmonic potential}\\
Let us consider the potential of the first family
$$V(z)=-z^2-2-\frac{8}{2z^2+1}+\frac{16}{(2z^2+1)^2}$$
The potential is analytic on $\mathbb{R}$, comes from the first case of Theorem \ref{thmmain} with the list $[4\nu+6]$, giving the gauge (of degree $0$ in $E$)
$$M(z,E)=-\frac{z^4-4\nu z^2+4\nu^2-4z^2+4\nu+1}{z(-z^2+2\nu+1)}$$
and then taking $\nu=-3/4$. The denominator of the expression of the eigenfunction (case $1$ Theorem \ref{thmmain0}) is $H=(E-3)z^2$, and thus the expression becomes singular for $E=3$.

The potential is analytic on $\mathbb{R}$, and thus so are the eigenfunctions. So the square integrability condition only put a condition near infinity. It is not trivial as the function $\mathcal{W}$ can be exponentially diverging at infinity. So let us first look at the co triangular condition of Proposition \ref{proptri}. The Stokes matrices at infinity of $\mathcal{W}$ are cotriangularizable if and only if $E\in 2\mathbb{Z}+1$. So we already know that the spectrum is a subset of this. It happens that this set leads to Liouvillian functions. We now compute the first eigenfunctions of this potential

\begin{align*}
\frac{e^{-z^2/2}}{2z^2+1}\quad & E=-1 & \frac{e^{-z^2/2}z(4z^4-5)}{2z^2+1}\quad & E=9\\
\frac{e^{-z^2/2}z(2z^2+3)}{2z^2+1}\quad & E=5 & \frac{e^{-z^2/2} (8z^6-12z^4-18z^2+3)}{2z^2+1}\quad & E=11\\
\frac{e^{-z^2/2}(4z^4+4z^2-1)}{2z^2+1}\quad & E=7 & \frac{e^{-z^2/2}z(8z^6-28z^4-14z^2+21)}{2z^2+1} \quad & E=13\\
\end{align*}
The polynomial appearing are in fact a linear combination with coefficients in $\mathbb{C}(z,E)$ of Hermite polynomials. Let us remark that the spectrum is similar to $-z^2$, except for few ``accidents'', $E=-1,1,3$. The accident $E=3$ is related to the singularity of the Gauge transformation $M$. In particular, as they are built, the gauge tranformation functions $M$ always have a particular behaviour at some specific points. However, if we evaluate the eigenfunction in $E$, the formula breaks downs for these particular $E$'s.\\

\begin{center}
\begin{figure}
\includegraphics[width=13.5cm]{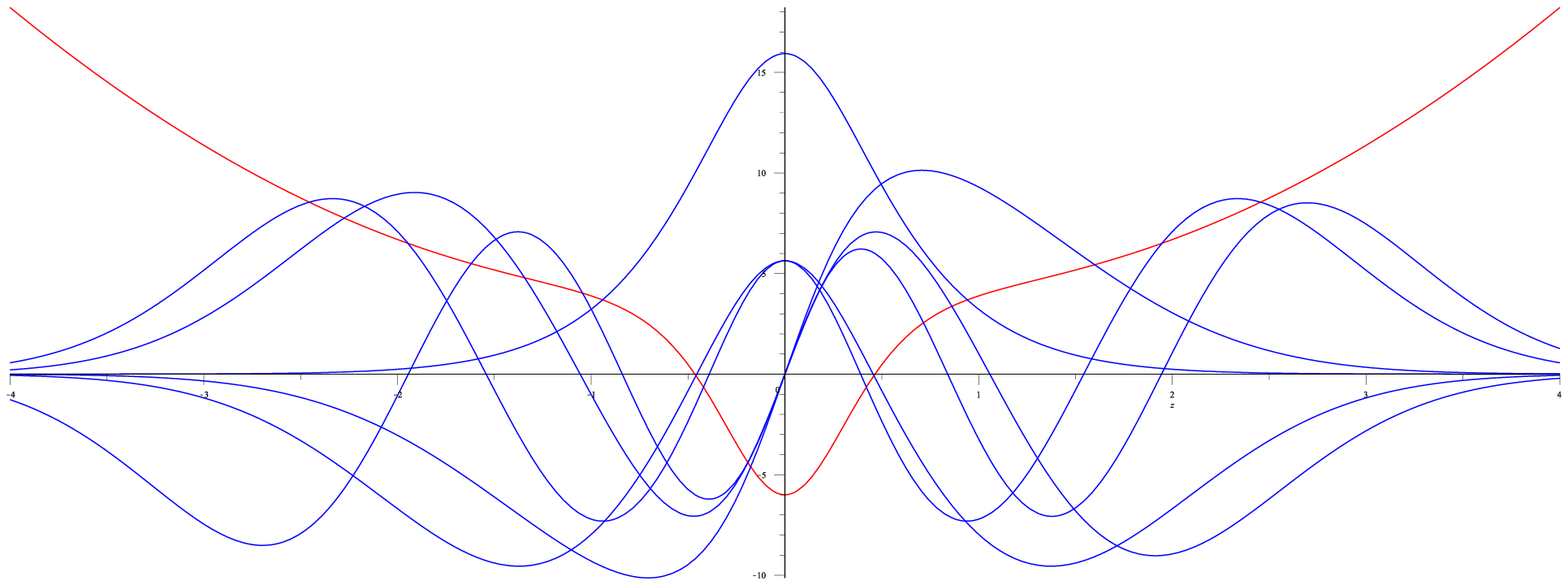}
\end{figure}
\end{center}

\noindent
\textbf{A fusion potential}

Let us consider a potential of the second family
$$V(z)=\frac{1}{z}-\frac{4}{z^2+2z+2}+\frac{8}{(z^2+2z+2)^2}$$
The potential is analytic on $\mathbb{R}^*$, comes from the second case of Theorem \ref{thmmain} with the list $[-1/(2\nu-1)^2,-1/(2\nu+2)^2]$, giving the gauge (of degree $1$ in $E$)
$$M(z,E)=\frac{-(2\nu+3)(2\nu-1)(8\nu^4+20\nu^3-8\nu^2z+6\nu^2-12\nu z+2z^2-9\nu)}{4z(4\nu^2+8\nu-2z+3)}E+$$
$$\frac{-8\nu^4-20\nu^3+8\nu^2z-30\nu^2+12\nu z-2z^2-31\nu+16z-6}{4z(4\nu^2+8\nu-2z+3)}$$
and then taking $\nu=-1/2$. The denominator of the expression of the eigenfunction (case $2$ Theorem \ref{thmmain0}) is 
$$H=\frac{(z^2+2z+2)^2(4E+1)^2}{4z^2},$$
and thus the expression becomes singular for $E=-1/4$. Remark this is a case of $w$ with a double root, the fusion occurring when taking $\nu=-1/2$.

The square integrability condition puts a condition near infinity and near $0$. There is problem on the interval definition of the solutions: the spectrum on $\mathbb{R}$, on $\mathbb{R}^+$ and $\mathbb{R}^-$ are not the same. The square integrability condition implies that at least $2$ of the three matrices involved (monodromy matrix at $0$ and $2$ Stokes matrices at infinity) should be cotriangularizable. As there is a multiplicative relation between these matrices, this implies that the differential Galois group is triangularizable. In other words, the condition $\mathcal{C}_1$ from Proposition \ref{proptri} are all the same. The condition $\mathcal{C}_1$ is then given by $E=-1/(4k^2),\; k\in\mathbb{N}^*$. The case $k=1$ leads to a square integrable solution on $\mathbb{R}^-$, the other ones on $\mathbb{R}^+$ (and none on $\mathbb{R}$).

\begin{align*}
\frac{e^{z/2}z}{z^2+2z+2}\quad & E=-1/4\\
\frac{e^{-z/4}z(z^3+6z^2+18z+24)}{z^2+2z+2}\quad & E=-1/16\\
\frac{e^{-z/6}z(z^4-4z^3-40z^2-144z-216)}{z^2+2z+2}\quad & E=-1/36\\
\frac{e^{-z/8}z(z^5-30z^4+50z^3+800z^2+3200z+5120)}{z^2+2z+2}\quad & E=-1/64\\
\end{align*}

\begin{center}
\begin{figure}
\includegraphics[width=13.5cm]{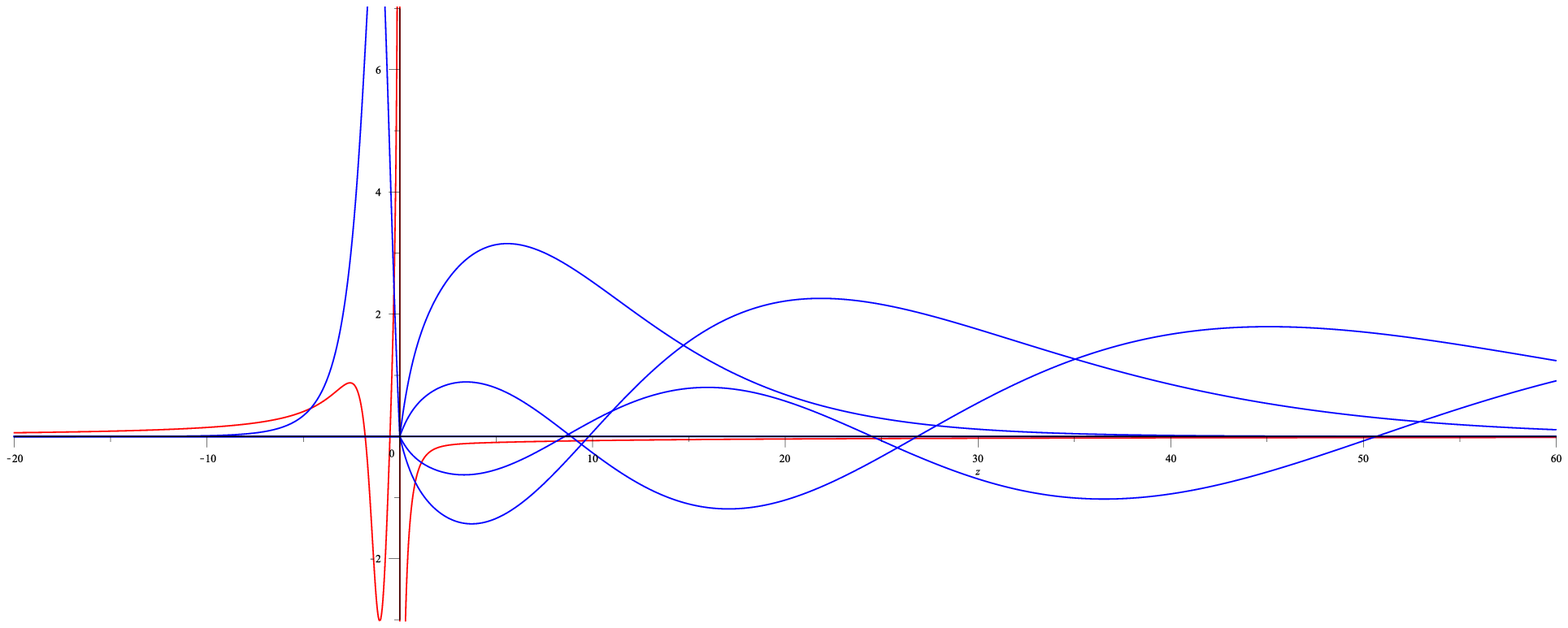}
\end{figure}
\end{center}

\noindent
\textbf{Continuous spectrum potentials}\\

The third case of Theorem \ref{thmmain0} gives two types of eigenfunctions, those with the Bessel function, and the Liouvillian ones. The potentials have continuous spectrum as their solutions are isomonodromic with respect to $E$. We give here two examples for low degree gauge functions $M$.\\

The fourth case of Theorem \ref{thmmain0} with $\nu=0$, $F(z)=\sqrt{z}(a+z^2+b\ln z)$, gives for the gauge
$$M(z,E)=\frac{2Ez^2+Eb-2}{4z}$$
The denominator of the expression of the eigenfunction (case $3$ Theorem \ref{thmmain0}) is 
$$H=\frac{1}{4}E^2(2z^2+b)^2,$$
giving here an example of a double root of $w$ at $E=0$. The corresponding potential is
$$V(z)=\frac{1}{4z^2}-\frac{8}{2z^2+b}+\frac{16b}{(2z^2+b)^2}$$
The eigenfunctions are given by case $3$ of Theorem \ref{thmmain0} with the $M$ given above.\\

The fifth case of Theorem \ref{thmmain0} with $\nu=1/2$, $F(z)=z^4+az^3+bz^2+cz+d$, gives for the gauge
$$M(z,E)=-\frac{3(4z+a)^2E}{(3a^2z+12az^2+16z^3+ab-2c)E-12a-48z}$$
The denominator of the expression of the eigenfunction (case $3$ Theorem \ref{thmmain0}) is 
$$H=\frac{4z^2(3a^2z+12az^2+16z^3+ab-2c)^2E^3}{(3Ea^2z+12Eaz^2+16Ez^3+Eab-2Ec-12a-48z)^2},$$
giving here an example of a triple root of $w$ at $E=0$. The corresponding potential is $V(z)=$
\begin{small}
$$\frac{-96z-24a}{3a^2z+12az^2+16z^3+ab-2c}-\frac{18a^4+72a^3z-72a^2b-288abz+144ac+576cz}{(3a^2z+12az^2+16z^3+ab-2c)^2}$$
\end{small}
and the (Liouvillian) eigenfunction
\begin{small}
$$\frac{(3a^2z+12az^2+16z^3+ab-2c)E+(3a^2+24az+48z^2)\sqrt{-E}-12a-48z}{3a^2z+12az^2+16z^3+ab-2c} e^{z\sqrt{-E}}$$
\end{small}

In these two cases, the monodromy and Stokes do not depend on $E$, and so any natural boundary condition are trivial (leading to a full $\mathbb{C}^*$ spectrum or empty spectrum). This implies by the way it is also the case for almost natural boundary conditions. 

\section{Conclusion}

We defined an explicit notion of quantum integrability for $1$ dimensional quantum system by building a differential field over $\mathbb{C}(z)$ with nice properties with respect to the monodromy/Stokes computations. All come down more or less to compute Gauge transformations of hypergeometric or confluent hypergeometric functions and hyperexponential functions. With these notions, we were able to completely classify integrable $1$ dimensional quantum problems in this sense. Remark that our classification effectively generates the integrable potentials, but does not answer the opposite question, i.e. given a potential, is it integrable? This can however be done using Theorem \ref{thmmain0}. Indeed, we need to search for rational gauge transformations of three particular differential equations. There is an algorithm in the Bessel case in \cite{16}, and probably soon for Whittaker functions. The quantum integrability would then be decidable in dimension $1$.

In the case of discrete spectrum, we always need one time or another to compute monodromy/Stokes matrices, as these appear in the boundary conditions and ``produce'' the spectrum. However, we do not have a complete understanding of the relation between the spectrum, in particular the gaps appearing in the examples and the singularities of the function $w$. Moreover, as many possible gauge functions $M$ are possible for one potential $V$, and we do not have a canonic gauge for a potential $V$, the roots of $w$ can depend on the choice of the gauge $M$.

In the continuous case, the monodromy/Stokes matrices do not play any role outside of being constant with respect to $E$. An important point is that the potentials obtained have always poles of order $2$, and this seems related to this continuous spectrum property. A natural question would be to ask if there are other systems for which the monodromy/Stokes matrices do not depend on $E$. Said otherwise, to find all rational functions $V\in\mathbb{C}(z)$ such that $\psi''(z)+(V(z)+E)\psi(z)=0$ is isomonodromic with respect to the parameter $E$. This problem is probably related to isomonodromic deformations and Painleve functions \cite{9}.

Another possibly way of generalization would be the higher dimensional case. Even in dimension $2$, the full classification is probably out of reach as we do not even know all ``classical'' integrable potentials. But still producing a definition of the same flavour would be interesting. In higher dimension, the notion of commutative observables become important: this is the analogue of the Liouville integrability of Hamiltonian system. So eigenfunctions are not anymore solution of a single PDE, but several. The natural condition to put on such a PDE system is holonomicity. Considering the characteristic variety associated to the corresponding differential ideal, we see that the holonomicity condition is in fact a dimension condition on this variety, and so similar to the independence conditions on first integrals in Hamiltonian systems. So what are the rigid functions solutions of a holonomic PDE system? The notions of differential Galois group can be generalized as we are still on some finite dimensional space. The notion of monodromy and Stokes are possibly more difficult. Hypergeometric functions have been generalized in many ways in higher dimensions, in particular $A$-hypergeometric functions. However, the possibility to carry an explicit computation of the monodromy is here of fundamental importance. Such result have not been yet obtained for $A$-hypergeometric functions.

\label{}

\bibliographystyle{model1-num-names}
\bibliography{quantum}

\begin{thebibliography}{21}
\expandafter\ifx\csname natexlab\endcsname\relax\def\natexlab#1{#1}\fi
\providecommand{\bibinfo}[2]{#2}
\ifx\xfnm\relax \def\xfnm[#1]{\unskip,\space#1}\fi
\bibitem[{Saad et~al.(2006)Saad, Hall, and Ciftci}]{3}
\bibinfo{author}{N.~Saad}, \bibinfo{author}{R.~L. Hall},
  \bibinfo{author}{H.~Ciftci},
\newblock \bibinfo{title}{Sextic anharmonic oscillators and orthogonal
  polynomials},
\newblock \bibinfo{journal}{Journal of Physics A: Mathematical and General}
  \bibinfo{volume}{39} (\bibinfo{year}{2006}) \bibinfo{pages}{8477}.
\bibitem[{Acosta-Hum{\'a}nez et~al.(2011)Acosta-Hum{\'a}nez, Morales-Ruiz, and
  Weil}]{1}
\bibinfo{author}{P.~B. Acosta-Hum{\'a}nez}, \bibinfo{author}{J.~J.
  Morales-Ruiz}, \bibinfo{author}{J.-A. Weil},
\newblock \bibinfo{title}{Galoisian approach to integrability of
  schr{\"o}dinger equation},
\newblock \bibinfo{journal}{Reports on Mathematical Physics}
  \bibinfo{volume}{67} (\bibinfo{year}{2011}) \bibinfo{pages}{305--374}.
\bibitem[{Kovacic(1986)}]{2}
\bibinfo{author}{J.~Kovacic},
\newblock \bibinfo{title}{An algorithm for solving second order linear
  homogeneous differential equations},
\newblock \bibinfo{journal}{Journal of Symbolic Computation}
  \bibinfo{volume}{2} (\bibinfo{year}{1986}) \bibinfo{pages}{3--43}.
\bibitem[{Van~Hoeij and Yuan(2010)}]{4}
\bibinfo{author}{M.~Van~Hoeij}, \bibinfo{author}{Q.~Yuan},
\newblock \bibinfo{title}{Finding all bessel type solutions for linear
  differential equations with rational function coefficients},
\newblock in: \bibinfo{booktitle}{Proceedings of the 2010 International
  Symposium on Symbolic and Algebraic Computation},
  \bibinfo{organization}{ACM}, pp. \bibinfo{pages}{37--44}.
\bibitem[{Bronstein and Lafaille(2002)}]{5}
\bibinfo{author}{M.~Bronstein}, \bibinfo{author}{S.~Lafaille},
\newblock \bibinfo{title}{Solutions of linear ordinary differential equations
  in terms of special functions},
\newblock in: \bibinfo{booktitle}{Proceedings of the 2002 international
  symposium on Symbolic and algebraic computation},
  \bibinfo{organization}{ACM}, pp. \bibinfo{pages}{23--28}.
\bibitem[{Kunwar and van Hoeij(2013)}]{6}
\bibinfo{author}{V.~J. Kunwar}, \bibinfo{author}{M.~van Hoeij},
\newblock \bibinfo{title}{Second order differential equations with
  hypergeometric solutions of degree three},
\newblock in: \bibinfo{booktitle}{Proceedings of the 38th International
  Symposium on Symbolic and Algebraic Computation},
  \bibinfo{organization}{ACM}, pp. \bibinfo{pages}{235--242}.
\bibitem[{Katz(2016)}]{8}
\bibinfo{author}{N.~M. Katz}, \bibinfo{title}{Rigid Local Systems.(AM-139)},
  volume \bibinfo{volume}{139}, \bibinfo{publisher}{Princeton University
  Press}, \bibinfo{year}{2016}.
\bibitem[{Slavyanov(2000)}]{9}
\bibinfo{author}{S.~Y. Slavyanov},
\newblock \bibinfo{title}{Isomonodromic deformations of heun and painleve?
  equations},
\newblock \bibinfo{journal}{Theoretical and Mathematical Physics}
  \bibinfo{volume}{123} (\bibinfo{year}{2000}) \bibinfo{pages}{744--753}.
\bibitem[{Simpson(1992)}]{14}
\bibinfo{author}{C.~T. Simpson},
\newblock \bibinfo{title}{Products of matrices},
\newblock in: \bibinfo{booktitle}{Differential Geometry, Global Analysis and
  Topology, Canadian Math. Soc. Conference Proceedings},
  volume~\bibinfo{volume}{12}, pp. \bibinfo{pages}{157--185}.
\bibitem[{Kostov(1999)}]{12}
\bibinfo{author}{V.~P. Kostov},
\newblock \bibinfo{title}{On the deligne-simpson problem},
\newblock \bibinfo{journal}{Comptes Rendus de l'Acad{\'e}mie des
  Sciences-Series I-Mathematics} \bibinfo{volume}{329} (\bibinfo{year}{1999})
  \bibinfo{pages}{657--662}.
\bibitem[{Van~Hoeij(2007)}]{15}
\bibinfo{author}{M.~Van~Hoeij},
\newblock \bibinfo{title}{Closed form solutions for linear differential and
  difference equations, project description}  (\bibinfo{year}{2007}).
\bibitem[{Ronveaux and Arscott(1995)}]{10}
\bibinfo{author}{A.~Ronveaux}, \bibinfo{author}{F.~M. Arscott},
  \bibinfo{title}{Heun's differential equations}, \bibinfo{publisher}{Clarendon
  Press}, \bibinfo{year}{1995}.
\bibitem[{Maier(2005)}]{11}
\bibinfo{author}{R.~S. Maier},
\newblock \bibinfo{title}{On reducing the heun equation to the hypergeometric
  equation},
\newblock \bibinfo{journal}{Journal of Differential Equations}
  \bibinfo{volume}{213} (\bibinfo{year}{2005}) \bibinfo{pages}{171--203}.
\bibitem[{Van~Hoeij and Vid{\=u}nas(2015)}]{7}
\bibinfo{author}{M.~Van~Hoeij}, \bibinfo{author}{R.~Vid{\=u}nas},
\newblock \bibinfo{title}{Belyi functions for hyperbolic hypergeometric-to-heun
  transformations},
\newblock \bibinfo{journal}{Journal of Algebra} \bibinfo{volume}{441}
  (\bibinfo{year}{2015}) \bibinfo{pages}{609--659}.
\bibitem[{Vid{\=u}nas(2009)}]{13}
\bibinfo{author}{R.~Vid{\=u}nas},
\newblock \bibinfo{title}{Algebraic transformations of gauss hypergeometric
  functions},
\newblock \bibinfo{journal}{Funkcialaj Ekvacioj} \bibinfo{volume}{52}
  (\bibinfo{year}{2009}) \bibinfo{pages}{139--180}.
\bibitem[{Van~Hoeij and Yuan(2010)}]{16}
\bibinfo{author}{M.~Van~Hoeij}, \bibinfo{author}{Q.~Yuan},
\newblock \bibinfo{title}{Finding all bessel type solutions for linear
  differential equations with rational function coefficients},
\newblock in: \bibinfo{booktitle}{Proceedings of the 2010 International
  Symposium on Symbolic and Algebraic Computation},
  \bibinfo{organization}{ACM}, pp. \bibinfo{pages}{37--44}.
\bibitem[{Hendriks and van~der Put(1995)}]{21}
\bibinfo{author}{P.~A. Hendriks}, \bibinfo{author}{M.~van~der Put},
\newblock \bibinfo{title}{Galois action on solutions of a differential
  equation},
\newblock \bibinfo{journal}{Journal of Symbolic Computation}
  \bibinfo{volume}{19} (\bibinfo{year}{1995}) \bibinfo{pages}{559--576}.
\bibitem[{Ruiz and Ruiz(1999)}]{17}
\bibinfo{author}{J.~J.~M. Ruiz}, \bibinfo{author}{J.~J.~M. Ruiz},
  \bibinfo{title}{Differential Galois theory and non-integrability of
  Hamiltonian systems}, \bibinfo{publisher}{Springer}, \bibinfo{year}{1999}.
\bibitem[{Chyzak and Salvy(1998)}]{18}
\bibinfo{author}{F.~Chyzak}, \bibinfo{author}{B.~Salvy},
\newblock \bibinfo{title}{Non-commutative elimination in ore algebras proves
  multivariate identities},
\newblock \bibinfo{journal}{Journal of Symbolic Computation}
  \bibinfo{volume}{26} (\bibinfo{year}{1998}) \bibinfo{pages}{187--227}.
\bibitem[{Chyzak(1994)}]{19}
\bibinfo{author}{F.~Chyzak},
\newblock \bibinfo{title}{Holonomic systems and automatic proofs of identities}
   (\bibinfo{year}{1994}).
\bibitem[{Koutschan(2010)}]{20}
\bibinfo{author}{C.~Koutschan},
\newblock \bibinfo{title}{Holonomic functions (user's guide)}
  (\bibinfo{year}{2010}).

\end{thebibliography}

\appendix
\section*{Appendix}

These are the Maple code used to generate each of the $4$ non trivial families of quantum integrable potentials of Theorem \ref{thmmain}. These codes are standalone, i.e. then can be copy/pasted directly into working programs.

\begin{scriptsize}
\begin{verbatim}

genpot1:=proc(L);
[seq(factor(-diff(ln(DETools[kovacicsols](
-z*diff(Y(z),z)+z^2*diff(Y(z),z,z)+Y(z)*(-z^4+L[i]*z^2-4*nu^2+1),Y(z))[1]),z)),i=1..nops(L))];
1/CurveFitting[RationalInterpolation]([seq([L[i],1/%[i]],i=1..nops(L))],E);
simplify(eval(subs(M(z)=%,z^(3/2)/sqrt(M(z)^2*z^2-z^4+z^2*E-diff(M(z),z)*z^2+M(z)*z-4*nu^2+1)*
(M(z)/(2*z)*W(E/4,nu,z^2)+D[3](W)(E/4,nu,z^2))))):
subs((D[3,3](W))(E/4,nu,z^2)=-(-1/4+E/4/z^2+(1/4-nu^2)/z^4)*W(E/4,nu,z^2),diff(%,z)):
convert(-factor(simplify(subs((D[3,3](W))(E/4,nu,z^2)=
-(-1/4+E/4/z^2+(1/4-nu^2)/z^4)*W(E/4,nu,z^2),diff(%,z)))/%%+E),parfrac,z);
end:

genpot2:=proc(L);
[seq(factor(-diff(ln(DETools[kovacicsols](
4*z^2*diff(Y(z),z,z)+(4*L[i]*z^2-4*nu^2+4*z+1)*Y(z),Y(z))[1]),z)),i=1..nops(L))];
piecewise(nops(L)=1,%[1],
1/CurveFitting[RationalInterpolation]([seq([L[i],1/%[i]],i=1..nops(L))],E)):
simplify(eval(subs(M(z)=%,z/sqrt(4*M(z)^2*z^2+4*E*z^2-4*(diff(M(z), z))*z^2-4*nu^2+4*z+1)*
(M(z)/sqrt(-4*E)*W(1/sqrt(-4*E),nu,z*sqrt(-4*E))+D[3](W)(1/sqrt(-4*E),nu,z*sqrt(-4*E))))));
subs((D[3,3](W))(1/sqrt(-4*E),nu,z*sqrt(-4*E))=
-((-1/4+1/(-4*E)/z+(1/4-nu^2)/z^2/(-4*E)))*W(1/sqrt(-4*E),nu,z*sqrt(-4*E)),diff(%,z)):
convert(-subs(gamma=sqrt(-E),factor(simplify(subs(E=-gamma^2,
subs((D[3,3](W))(1/sqrt(-4*E),nu,z*sqrt(-4*E))=
-((-1/4+1/(-4*E)/z+(1/4-nu^2)/z^2/(-4*E)))*W(1/sqrt(-4*E),nu,z*sqrt(-4*E)),
diff(%,z))/%%)+E))),parfrac,z);
end:

genpot3:=proc(F) local n,i,S;
n:=degree(coeff(expand(F/sqrt(z)),ln(z),0),z)/2+1:
S:=E^(n-1)*F: for i from 1 to n-1 do S:=(-diff(S,z,z)-1/(4*z^2)*S)/E+E^(n-1)*F: od:
numapprox[pade](series(-diff(ln(collect(S,E,factor)),z),E=0,n),E,[floor(n/2),floor((n-1)/2)]);
subs(M(z)=%,z/sqrt(4*M(z)^2*z^2+4*E*z^2-4*(diff(M(z), z))*z^2-4*(0)^2+1)*
(M(z)/sqrt(-4*E)*W(0,0,z*sqrt(-4*E))+D[3](W)(0,0,z*sqrt(-4*E))));
subs((D[3,3](W))(0,0,z*sqrt(-4*E))=-(-1/4+1/4/(z^2*(-4*E)))*W(0,0,z*sqrt(-4*E)),diff(%,z)):
convert(-subs(gamma=sqrt(-E),factor(simplify(subs(E=-gamma^2,
subs((D[3,3](W))(0,0,z*sqrt(-4*E))=
-(-1/4+1/4/(z^2*(-4*E)))*W(0,0,z*sqrt(-4*E)),diff(%,z))/%%)+E))),parfrac,z);
end:

genpot4:=proc(F) local n;
n:=degree(F)/2+1:
numapprox[pade](-diff(ln(add((-1)^i*piecewise(i=0,F,diff(F,z$(2*i)))*E^(n-1-i),i=0..n-1)),z)
,E,[floor(n/2),floor((n-1)/2)]);
simplify(eval(subs(M(z)=%,z/sqrt(4*M(z)^2*z^2+4*E*z^2-4*(diff(M(z), z))*z^2-4*(1/2)^2+1)*
(M(z)/sqrt(-4*E)*W(0,1/2,z*sqrt(-4*E))+D[3](W)(0,1/2,z*sqrt(-4*E))))));
subs((D[3,3](W))(0,1/2,z*sqrt(-4*E))=1/4*W(0,1/2,z*sqrt(-4*E)),diff(%,z)):
convert(-subs(gamma=sqrt(-E),factor(simplify(subs(E=-gamma^2,
subs((D[3,3](W))(0,1/2,z*sqrt(-4*E))=1/4*W(0,1/2,z*sqrt(-4*E)),diff(%,z))/%%)+E))),parfrac,z);
end:

\end{verbatim}
\end{scriptsize}

The input for the two first ones are lists of the form required by Theorem \ref{thmmain}, and for the last two are functions, the function $F$ of the form required by Theorem \ref{thmmain}. Remark moreover that these last two can handle functions $F$ with parameters. Some implementation tricks have been used
\begin{itemize}
\item The $\!_1F_1$ functions of Theorem \ref{thmmain} are generated on the fly by the Kovacic algorithm
\item The rational interpolation is made on the function $1/M$, as the default degrees of the interpolation algorithm then meet with the requirements of Theorem \ref{thmmain}
\item In the third one, the series defining $M$ is generated recursively through iterated application of differential operator $D$
\item The substitution of $E$ by $-\gamma^2$ is used to force simplifications of the square roots of $E$ (and choose the same valuation for all of them), so that the resulting potential can be put under partial decomposition form.
\end{itemize}

\end{document}